\documentclass[journal]{IEEEtran}
\ifCLASSINFOpdf
\else
\fi
\usepackage{animate}
\usepackage{amsmath,amssymb,amsthm}
\usepackage{graphicx,subfigure,epstopdf}
\usepackage[comma,numbers,square,sort&compress]{natbib}
\usepackage{epstopdf}
\usepackage{epic}
\usepackage{hyperref}
\usepackage{subfigure}
\usepackage{graphicx}
\usepackage{epstopdf}
\usepackage{epsfig}
\usepackage{arydshln}
\usepackage{enumerate}
\usepackage[utf8]{inputenc}
\usepackage[english]{babel}
\usepackage{fancyhdr}
\usepackage{float}
\usepackage{lastpage}
\usepackage{color}
\usepackage{dsfont}
\usepackage{tikz}
\usepackage{ifthen}
\RequirePackage{doi}
\usepackage{hyperref}
\usetikzlibrary{positioning}
\usepackage{changes}
\usetikzlibrary{shapes,snakes}
\usepackage{verbatim}
\usepackage{ifthen}
\usepackage{pgf}
\usepackage{indentfirst}
\usepackage{todonotes}

\usetikzlibrary{shapes.geometric}

\newcount\mycount
\usetikzlibrary[topaths]

\newcommand{\EQ}{\begin{eqnarray}}
\newcommand{\EN}{\end{eqnarray}}
\newcommand{\EQQ}{\begin{eqnarray*}}
\newcommand{\ENN}{\end{eqnarray*}}
\newcommand{\col}{\mbox{col}}
\newcommand{\diag}{\mbox{diag }}

\newtheorem{thm}{Theorem}

\newtheorem{lem}{Lemma}
\newtheorem{rem}{Remark}
\newtheorem{defi}{Definition}
\newtheorem{prop}{\it Property}
\newtheorem{prob}{Problem}
\newtheorem{exam}{Example}
\newtheorem{ass}{Assumption}

\newtheorem{claim}{\bf \em{Claim}}

\newcommand{\myr}{\color{red}}
\allowdisplaybreaks
\begin{document}

\title{Adaptive Cooperative Tracking and Parameter Estimation of an Uncertain Leader over General Directed Graphs}
\author{ Shimin~Wang, ~Hongwei~Zhang, and~Zhiyong~Chen
\thanks{This work was supported by the National Natural Science Foundation of China under Projects 61773322 and 51729501. (Corresponding author: Hongwei Zhang) }
\thanks{Shimin Wang is with Department of Electrical and Computer Engineering, University
of Alberta, Edmonton, Alberta, Canada. E-mail: shimin1@ualberta.ca. }
\thanks{Hongwei Zhang is with the School of Mechanical Engineering and Automation, Harbin Institute of Technology, Shenzhen, Guangdong 518055, P.R. China. E-mail: hwzhang@hit.edu.cn.}
\thanks{Zhiyong Chen is with the School of Electrical Engineering and Computing, The University of Newcastle, Callaghan, NSW 2308, Australia. E-mail: zhiyong.chen@newcastle.edu.au.
}
}

\maketitle

\begin{abstract}
This paper studies cooperative tracking problem of heterogeneous Euler-Lagrange systems with an uncertain leader. Different from most existing works, system dynamic knowledge of the leader node is unaccessible to any follower node in our paper. Distributed adaptive observers are designed for all follower nodes, simultaneously estimate the state and parameters of the leader node.
The observer design does not rely on the frequency knowledge of the leader node, and the estimation errors are shown to converge to zero exponentially.
Moreover, the results are applied to general directed graphs, where the symmetry of Laplacian matrix does not hold. This is due to two newly developed Lyapunov equations, which solely depend on communication network topologies.
Interestingly, using these Lyapunov equations, many results of multi-agent systems over undirected graphs can be extended to general directed graphs.
Finally, this paper also advances the knowledge base of adaptive control systems by providing a main tool in the analysis of parameter convergence for adaptive observers.

\end{abstract}


\section{INTRODUCTION} \label{Sect-Intro}

 Due to its widespread applications in robotics, unmanned aerial vehicles, smart grids, wireless sensor networks, social networks, etc, cooperative control of multi-agent systems (MASs) has been receiving tremendous attentions in the control community since the early 2000s.
  For a comprehensive literature review, readers are referred to some recent survey papers \cite{CaoYuRenChen2013Review,KnornChenMiddleton2016Overview,OhaParkAhn2015Survey,xiao2017adaptive} and references therein.

Two fundamental problems of cooperative control of MASs are leaderless consensus problem and leader-following consensus problem, or  known as cooperative tracking problem \cite{ZhangLewis2011TAC-optimal}. For the former one, all agents have equal roles and  achieve consensus to a common trajectory, which depends on the initial values of the agents. This paper is concerned about the latter one where a leader node, acting as a command generator, generates a trajectory for all follower nodes to track. The challenge for cooperative tracking of MASs lies in the fact that each follower agent can only obtain local information from its neighbors, instead of information from all agents, and only partial agents can directly access the leader's information. In this sense, a distributed control law is required and the system analysis involves both agent dynamics and communication topology.

Although cooperative tracking of homogeneous MASs, with identical system dynamics for both follower nodes and the leader node, is relatively less challenging to design and analysis, heterogeneous MASs are more practical and general. For example, many industrial manufacturing processes rely on cooperation of different types of robotic manipulators.
 For cooperative tracking of  heterogeneous linear MASs, distributed observer design approach was proposed in \cite{SuHuang2012TAC-coop}, where each follower agent maintains an observer, estimating the state information of the leader agent.
 However, all observers are required to know the leader's dynamics, i.e., the system matrix. The same requirement is also assumed in \cite{LiZK2016Auto-distri}. Noting that the above condition implicitly implies direct communication between each follower node and the leader node, which violates the distributed nature of MASs.  The observer design approach was later extended by \cite{CaiHuang2016TAC-EL},
  where only those follower nodes who have directed links from the leader node need to know the leader's dynamics. Although this observer design is distributed, it is however still impractical in most engineering applications. For example, when the leader system  generates a sinusoidal signal, knowing the leader's dynamics amounts to knowing the frequency of the leader trajectory, and this is often challenging.
It is well known that frequency estimation problem itself has long been a classical topic in control community \cite{Pesun1996,HsuOrtega1999TAC-globally,Dochain2003,Astolfi2014TAC-semi}, and has many practical applications, such as active noise and vibration control \cite{fuller1995CSM-active} in helicopters and disk drives, and autopilot control of an autonomous air vehicle for vertical landing on a deck oscillating in the vertical direction due to high sea states \cite{marconi2002autonomous}.
When the leader's dynamics is unaccessible to all follower nodes, it is known as a knowledge-based leader \cite{wang2006theoretical} or an uncertain leader \citep{wanghang2019TNN,baldi2020distributed}.

 Cooperative tracking of heterogeneous MASs with an uncertain leader is an important and challenging problem and has attracted increasing {attentions} in recent years \cite{MHNSPLBL2016,wu2017Auto-adaptive,yan2018new,wanghang2019TNN,WangHuang2019IJC-coop,LvLiu2019TAC-leader,jiaoTrentelman2019,WangMeng2021TAC-adaptive,baldi2020distributed}. Particularly, \cite{MHNSPLBL2016} proposed a distributed dynamic compensator for the uncertain leader and
 showed that the compensator can estimate the state of the leader system asymptotically using off-policy reinforcement learning. However, \cite{MHNSPLBL2016} did not consider the convergence issue of the estimated parameter of the system matrix. In \cite{wu2017Auto-adaptive}, a distributed adaptive reference generator was designed for each follower node to estimate the leader's system parameters and states, and the convergence of estimation errors was shown to be exponential. But the uncertain parameters of the leader node are restricted to be within some known compact set  and the output matrix is in a special form. In \cite{wanghang2019TNN}, the authors developed an adaptive distributed observer that is able to estimate both  the state and unknown dynamics parameters of the leader system, and showed that the estimated parameters can converge to the actual values asymptotically provided that the state of the leader system is persistently exciting. The convergence of the observer was later shown to be exponential in \cite{WangHuang2019IJC-coop}. The work \cite{LvLiu2019TAC-leader} also studied leader-following consensus problem with an unknown leader, and designed a distributed dynamic compensator to deal with the unknown parameters in the leader system. Very recently, using only the output information of the leader, \cite{WangMeng2021TAC-adaptive} solved  the leader-following consensus and parameter estimation problems simultaneously, where the leader signal is a sum of sinusoids with unknown amplitudes, frequencies and initial phases. In addition, \cite{baldi2020distributed} also designed an adaptive distributed observer for an uncertain leader to solve leader-following consensus and parameter estimation problems over undirected graph.

It is worth mentioning that all the above mentioned works require the communication network among follower nodes  to be either undirected \cite{wang2006theoretical,wu2017Auto-adaptive,wanghang2019TNN,WangHuang2019IJC-coop,baldi2020distributed,LvLiu2019TAC-leader,WangMeng2021TAC-adaptive} or detailed balanced \cite{MHNSPLBL2016}. For undirected graph, the adjacency matrix and Laplacian matrix are all symmetric, leading to a relatively easier treatment for the system analysis. While for a detailed balanced graph, although its adjacency/Laplacian matrix is generally not symmetric, there exists a diagonal matrix, such that the product of this diagonal matrix and adjacency/Laplacian matrix is symmetric.  Both undirected graph and detailed balanced digraph imply full duplex communication network. However, simplex communication network, modeled by unidirectional directed graph, is often preferred when an energy efficiency and cost effective design is required, such as formation of battery powered unmanned aerial vehicles. {Moreover,} in some scenarios, different agents may {be} equipped with different sensors, and different sensing {radiuses} will result in directed communication graphs. Therefore, compared with undirected graphs, cooperative control of MASs over general directed graphs is considered to be much more practical in literature \cite{ZhangHW2014}. Yet, the challenge of dealing with directed graphs has been well recognized due to the asymmetry of the Laplacian matrix \cite{ZhangHW2014}. Although the observer in \cite{MHNSPLBL2016} was also applied to general directed graphs, the convergence is only shown to be uniformly ultimately bounded. The results in \cite{wanghang2019TNN} were further extended to directed graphs in a very recent work \cite{WangHuang2020IJC-adapt}, but the directed graph is restricted to be acyclic, i.e., the Laplacian matrix of the graph is in a lower triangular form.

Based on the above mentioned statements, this paper aims to solve the cooperative tracking problem of heterogeneous systems with an uncertain leader over general directed graphs, with emphasis on simultaneous estimation of the state and parameters of the leader system. The leader system can generate a linear combination of muliti-tone sinusoidal signals with unknown frequencies. The follower agents take the Euler-Lagrange (EL) dynamics \cite{lewis2004Book-robot}, which can describe the behaviour of a large class of engineering systems, such as mechanical systems.
It is worth mentioning that due to the separation principle, the distributed observer designed in this paper can also be used to solve more classes of cooperative tracking problems of MASs other than EL systems. The main contributions of this paper are summarized as follows.
\begin{enumerate}
\item It solves a new distributed parameter estimation problem of multi-agent systems with an uncertain leader over general directed graphs and shows that this observer design can be applied to many classes of cooperative tracking problems, such as multiple EL systems. The technical challenges lie in the facts that the leader's dynamic knowledge is unaccessible to any follower node and the adjacency/Laplacian matrix of the communication network is not symmetric.
\item  A new tool for analysing parameter convergence of adaptive observers is developed, which enriches design tools in adaptive control systems and is a complement of the fundamental result \cite[Lemma B.2.3]{marino1996nonlinear}.
\item Also, it brings two new design tools for cooperative control of MASs over directed graphs, i.e., two graph based Lyapunov equations. By using these Lyapunov equations, many results of MASs over undirected graphs can be extended to directed graphs.
\end{enumerate}
The rest of this paper is organized as follows. Section \ref{Sect-preli} presents some preliminaries of graph theory and matrix theory, and formulates the problem. A distributed adaptive observer is designed and analyzed in Section \ref{Sect-observer}, and a control law is proposed in Section \ref{Sect-controller}. A simulation example is given in Section \ref{Sect-example} to illustrate the efficacy of the proposed algorithms, and
Section \ref{Sect-Con} concludes the paper.

\section{Preliminaries and problem formulation} \label{Sect-preli}
This section introduces notations, some background of graph theory and M-matrix, and formulates the cooperative tracking problem.

\subsection{Notations}\label{Subsect-notation}
The notations used throughout this paper is rather standard. The identity matrix with appropriate dimensions is $I$, while $I_l \in \mathds{R}^{l \times l}$ means an {identity} matrix with specific dimensions. For a vector $ x\in \mathds{R}^n$, $x_i$ denotes the $i$-th component of $x$, and it can be denoted by $x=[x_i]$. Similarly, for a matrix $A=[a_{ij}]\in \mathds{R}^{n\times n}$, $a_{ij}$ is its $(i,j)$-th entry. Matrix $A$ is called positive (nonnegative), denoted by $A \succ 0$ ($A \succeq 0$), if all its entries $a_{ij}$ are positive (nonnegative). Matrix $A>0$ ($A \geq 0$) means it is positive definite (positive semidefinite).  The spectral radius of matrix $A$ is denoted by $\varrho(A)$. For any matrix $A$, $\sigma_{\min}(A)$ and $\sigma_{\max}(A)$ denote the minimal and maximal singular values of $A$, respectively.
 The Kronecker product is denoted by $\otimes$. Let $\|\cdot\|$ denote both Euclidian norm of a vector and induced 2-norm of a matrix, i.e., $\|A\|=\max_{\|x\|=1}\|Ax\|$, and $\|\cdot\|_F$ denote the Frobenius norm of a matrix.  For $X_1,\cdots,X_k\in \mathds{R}^{n\times m}$, let $ \mbox{col }(X_1,\cdots,X_k)=\left[ X_{1}^T,\cdots,X_{k}^T\right]^T.$  Denote $\mathds{1}_n= \col(1,\cdots,1)\in \mathds{R}^{n}$.
 For matrix $S = [s_1, s_2, \dots, s_n] \in \mathds{R}^{m\times n}$ with $s_i\in \mathds{R}^{m}$, $\hbox{vec}(S)=\col\left(s_1,\cdots,s_n\right) \in \mathds{R}^{mn}.$
 A diagonal matrix is denoted by
$$
  \textnormal{block diag}(\gamma_1,\dots,\gamma_n) =\left[
                 \begin{matrix}
                 \gamma_1 &\cdots  &  0 \\
                 \vdots& \ddots &  \vdots \\
                 0&  \cdots & \gamma_n\\
                 \end{matrix}
                 \right],
$$
Let $\diag(\cdot)$ be such a function that
\begin{align}
  \diag(x) &=\textnormal{block diag}(x_1,\dots,x_n).\nonumber
\end{align}
 For  $x = \col(x_1, x_2, \dots, x_{2l})\in \mathds{R}^{2l}$, define
\begin{equation}\label{phimatrix}
  \phi(x)=\left[
                           \begin{array}{cccc}
                              -x_{2} &   x_{1}  & \cdots &  0  \\
                              \vdots &  \ddots & \ddots &  \vdots \\
                             0  &  \cdots & -x_{2l} &  x_{2l-1} \\
                           \end{array}
                         \right].
\end{equation}

\subsection{Graph theory}\label{Subsect-Graph}

For cooperative tracking problem, follower nodes collaborate with each other through a communication network, which can be modeled by a graph $\mathcal{G}(\mathcal{V}, \mathcal{E})$, where $\mathcal{V}=\{1,2,\dots,N\}$ is the node set and $\mathcal{E}= \mathcal{V} \times \mathcal{V}$ is the edge set. An edge from node $i$ to node $j$ can be denoted by an ordered pair $(i,j)$.  We write $(i,j)\in \mathcal{E}$, if there is a directed edge from node $i$ to node $j$. The topology of the graph $\mathcal{G}$ can be fully captured by the adjacency matrix $\mathcal{A}=[a_{ij}]\in \mathds{R}^{N \times N}$, where $a_{ij}>0$ if $(j,i)\in \mathcal{E}$, and $a_{ij}=0$ otherwise. If $a_{ij}>0$, then node $j$ is called a neighbor of node $i$. Denote $\mathcal{N}_i= \{j~|~a_{ij}>0,~j=1,2,\dots,N\}$ as the neighbor set of node $i$.
The Laplacian matrix of the graph $\mathcal{G}$ is defined as $L=[l_{ij}]\in \mathds{R}^{N \times N}$, with $l_{ii}=\sum_{j=1}^N a_{ij}$ and $l_{ij}=-a_{ij}$ for $i \neq j$. Label the leader node as node $0$ and define an augmented graph consisting both leader node and all follower nodes, along with all corresponding edges, i.e., $\bar{\mathcal{G}}(\bar{\mathcal{V}}, \bar{\mathcal{E}})$, where $\bar{\mathcal{V}}=\{0, 1,2,\dots,N\}$ and $\bar{\mathcal{E}}=\bar{\mathcal{V}} \times \bar{\mathcal{V}}$. If there is a directed edge from the leader node to node $i$, then $a_{i0}=1$; otherwise $a_{i0}=0$. Denote $\bar{\mathcal{N}}_i= \{j~|~a_{ij}>0,~j=0,1,2,\dots,N\}$.  A graph is undirected, if $a_{ij}=a_{ji}$; otherwise, it is directed. A directed path is a sequence of edges $\{(i,j), (j,l), (l, m), \dots\}$  in a directed graph. An undirected path is defined similarly. An undirected graph is connected if there exists an undirected path between any two distinct nodes. A directed graph has a spanning tree, if there is a root node which have no incoming edges, and there is a directed path from the root node to every other node. A directed graph is said to be detailed balanced if there exist some positive real numbers $k_i$, such that $k_i a_{ij} = k_j a_{ji}$ for all $i,j$.

\subsection{Matrix theory}\label{Subsect-Matrix}
\begin{defi}\cite{berman1994nonnegative} \label{Def-Mmatrix}
Matrix $A_M \in \mathbb{R}^{n\times n}$ is a nonsingular $M$-matrix, if it can be expressed in the following form
$$A_M=\beta_M I - B_M, \quad \beta_M > \varrho(B_M), \quad B_M \succeq 0.$$
If $\beta_M = \varrho(B_M)$, then $A_M$ is called a singular $M$-matrix.
\end{defi}
By definition, if $A_M=\left[a_{M,ij}\right]$ is a nonsingular $M$-matrix, then it has the following two properties:
\begin{enumerate}
 \item $a_{M,ii}>0, ~\forall i$, $a_{M,ij}\leq 0, \forall i\neq j$;
 \item all its eigenvalues have positive real parts.
\end{enumerate}
\begin{lem} \cite[Theorem 6.2.3 \& 6.2.7]{berman1994nonnegative} \label{nsMmatrix}
If $A_M$ is a nonsingular $M$-matrix, then $A_M^{-1}\succeq 0$. Moreover, if $A_M$ is irreducible, then $A_M^{-1}\succ 0$.
\end{lem}

\begin{lem} \cite{Fisherfuller1958,Hershkowitz1992}\label{dstable} Suppose $A_M \in \mathds{R}^{n \times n}$ has positive leading principle minors. Then there exists a positive diagonal matrix $D_M$ such that $D_MA_M$ has simple positive eigenvalues.
\end{lem}

\begin{lem} \cite[Lemma 1]{Barkana2006Auto-mitigation} \label{matWC}
Consider matrices $W_M\in \mathds{R}^{m\times m}$ and $C_M\in \mathds{R}^{m\times m}$. If $C_M$ is diagonalizable
and has real and positive eigenvalues, then
$$W_MC_M=C^T_M W_M$$
has a symmetric positive definite solution for $W_M$ that makes $W_MC_M$ symmetric and positive definition.
\end{lem}

\subsection{Problem formulation}\label{Subsect-Prob}

Consider a group of $N$ follower agents, modeled by heterogeneous Euler-Lagrange dynamics \cite{lewis2004Book-robot}:
\begin{equation}\label{Model-follower}
  M_i(q_i)\ddot{q}_i+C_i(q_i,\dot{q}_i)\dot{q}_i+ G_i(q_i)=\tau_i, \quad i=1,2,\dots,N,
\end{equation}
where $q_i\in \mathds{R}^n$ is the vector of generalized coordinates,
$M_i\left(q_i\right)\in \mathds{R}^{n\times n}$ is the inertia matrix, $C_i\left(q_i,\dot{q}_i\right)\in \mathds{R}^{n\times n}$ is the {vector of} Coriolis and centripetal forces, ${G_i}\left(q_i\right)\in \mathds{R}^{n}$ is the vector of gravitational force, and $\tau_i\in \mathds{R}^{n}$ is the control torque.

The dynamics \eqref{Model-follower} satisfy the following three properties \cite{lewis2004Book-robot}:
\begin{prop}
$M_i(q_i)$ is symmetric and positive definite.
\end{prop}

\begin{prop}\label{prop1}
$M_i(q_i)\ddot{q}_i+C_i(q_i,\dot{q}_i)\dot{q}_i+G_i(q_i)=Y_i(q_i,\dot{q}_i,\ddot{q}_i)\Theta_i$, where $Y_i(q_i,\dot{q}_i,\ddot{q}_i)\in \mathds{R}^{n\times q}$ is a known regression matrix and $\Theta_i\in \mathds{R}^{q}$ is a constant vector consisting of the parameters of \eqref{Model-follower}.
\end{prop}

\begin{prop}
$\dot{M}_i\left(q_i\right)-2C_i\left(q_i,\dot{q}_i\right)$ is skew symmetric, where $\dot{M}_i\left(q_i\right)$ is the time derivative of $M_i(q_i)$.
\end{prop}
The leader agent is described by the following linear system
\begin{subequations}\label{leader}
\begin{align}
 \dot{v}&=S(\omega)v,\label{leadera}\\
y&=Ev,\label{leaderb}
 \end{align}
 \end{subequations}
where $v\in \mathds{R}^{m}$ is the state, $y\in \mathds{R}^{n}$ is the output, $S(\omega)\in\mathds{R}^{m\times m}$ is the system matrix with an unknown parameter $\omega\in \mathds{R}^{l}$, and $E\in \mathds{R}^{n\times m}$ is an unknown output matrix. Suppose $S(\omega)$ satisfies the following assumption.

\begin{ass} \label{Ass-S}
All the eigenvalues of $S(\omega)$ are simple with zero real parts.
\end{ass}
This assumption is rather standard in the literature of cooperative tracking of MASs \cite{CaiHuang2016TAC-EL,wu2017Auto-adaptive, wanghang2019TNN}.  Under Assumption \ref{Ass-S}, the leader agent \eqref{leader} can generate multi-tone sinusoidal signals, an important class of signals and frequently encountered in industry, such as vibration signals of a rotating machinery \cite{fuller1995CSM-active} and vertical oscillating motion of a ship due to high sea states\cite{marconi2002autonomous}. See also \cite{HsuOrtega1999TAC-globally, Astolfi2014TAC-semi} and references therein for more examples. Since both $\omega$ and $E$ are unknown, the generated multi-tone sinusoidal signals may have unknown frequencies, amplitudes, and initial phases.

Under Assumption \ref{Ass-S},  we can assume
\begin{equation}\label{skew}
 S(\omega) = \diag(\omega)\otimes a,
 \end{equation}
where $a=\left[
          \begin{smallmatrix}
            0& 1 \\
            -1 & 0 \\
          \end{smallmatrix}
        \right]$,
$\omega=\textnormal{\col}\left(\omega_{01},\cdots,\omega_{0l} \right)\in \mathds{R}^{l}$, $l=m/2$.  Obviously, $S(\omega)$ is skew-symmetric.

This paper aims to solve the following cooperative tracking problem of heterogeneous EL systems with emphasis on simultaneous estimation of the state $v$, and the parameters $\omega$ and $E$ of the leader system \eqref{leader}.

\begin{prob}\label{Problem-track} (Cooperative tracking problem)
Consider a MAS consisting of $N$ follower nodes \eqref{Model-follower} and a leader node \eqref{leader}. Given two compact sets $\mathds{V}_0\subset \mathds{R}^{m}$ and $\mathds{W}\subset \mathds{R}^{l}$ containing the origins, design distributed control law $\tau_i$ for each follower node $i$, such that for any $v(0)\in \mathds{V}_0$, $q_i(0)\in \mathds{R}^{n}$, $\dot q_i(0)\in \mathds{R}^{n}$ and $\omega\in \mathds{W}$,
 the solution of the closed-loop system is uniformly bounded for all $t\geq0$ and
 $$ \lim_{t \to \infty}(q_i(t) - y(t)) =0,~~~~i=1,\cdots,N.$$
\end{prob}

To solve this problem, we need the following assumption on topology of the augmented communication graph $\bar{\mathcal{G}}$, which was recognized as the most mild condition on the communication graph in the literature of cooperative tracking of  MASs \cite{ZhangLewis2012Auto-adaptive}.

\begin{ass} \label{Ass-graph}
The augmented graph $\bar{\mathcal{G}}$ has a spanning tree with the leader node being the root.
\end{ass}

Define $G = \textnormal{block diag}(a_{10}, a_{20}, \dots, a_{N0})$ and $$H=L+G.$$ Then under Assumption \ref{Ass-graph}, $H$ is nonsingular and all its eigenvalues have positive real parts \cite{HuHong2007PA-leader}.

\section{Adaptive observer design and analysis} \label{Sect-observer}
To solve Problem \ref{Problem-track}, we first need to estimate the unknown parameters $\omega$ and $E$, and the state $v$ of the leader agent \eqref{leader}, which is the objective of this section.

\subsection{Observer design}
A distributed adaptive observer for each follower node $i$ $(i=1,2,\dots,N)$ is designed as
\begin{subequations}\label{compensator2}
\begin{align}
\dot{\hat{\eta}}_i &=S\left(\hat{\omega}_i \right)\hat{\eta}_i + \mu_1 d_i\sum\nolimits_{j \in \mathcal{\bar{N}}_i} a_{ij}(\hat{\eta}_j-\hat{\eta}_i), \label{compensator2b}\\
\dot{\hat{\omega}}_{i} &=\mu_2 d_i \phi \left( \sum\nolimits_{j \in \mathcal{\bar{N}}_i}a_{ij}(\hat{\eta}_j-\hat{\eta}_i )\right )\hat{\eta}_i,\label{compensator2c}\\
\dot{\hat{E}}_i &=d_i\sum\nolimits_{j\in \mathcal{\bar{N}}_i}a_{ij}(\hat{y}_j-\hat{y}_i)\hat{\eta}_i^{T}, \label{compensator2-E}
\end{align}
\end{subequations}
where $\hat{\eta}_i$, $\hat{\omega}_i$, and $\Hat{E}_i$ are the estimated values of $v$, $\omega$, and $E$, respectively,
$\mu_1>0$, $\mu_2>0$, and $d_i>0$ are the design parameters to be specified later,  $\hat{y}_i=\hat{E}_i \hat{\eta}_i$, and $\hat{\eta}_0=v$, $\hat{y}_0=y$.

Let the estimation errors be $\tilde{\eta}_i = \hat{\eta}_i - v$, $\tilde{\omega}_i = \hat{\omega}_i - \omega$, and $\tilde{E}_i = \hat{E}_i - E$. In the following development, we shall show how to design $\mu_1$, $\mu_2$, and $d_i$ such that
\begin{subequations}
\begin{align}
 &\lim_{t \to \infty} \tilde{\eta}_i(t) = 0, \label{des1}\\
& \lim_{t \rightarrow \infty} \tilde{\omega}_i(t) = 0, \label{des2}\\
& \lim_{t \rightarrow \infty} \tilde{E}_i(t) = 0. \label{des3}
\end{align}
\end{subequations}

\begin{rem}
It was shown that, for the detailed balanced graph, the dynamic compensator proposed in \cite{MHNSPLBL2016} can also estimate the state of the leader \eqref{leader} asymptotically. However, \cite{MHNSPLBL2016} did not consider the convergence issue of the estimation of the leader's dynamics $S$.
In contrast to the observer in \cite{CaiHuang2016TAC-EL}, the information of the leader's frequency is not required by any observer, as can be observed from \eqref{compensator2}.
Another advantages of the adaptive distributed observer \eqref{compensator2} is that it estimates only $\frac{m}{2}$ unknown parameters of $S$ by $\frac{mN}{2}$ equations rather than all the $m\times m$ entries of $S$ by $m^2N$ equations as in \cite{MHNSPLBL2016}. Therefore, the computational cost is significantly reduced.
Synchronization of heterogeneous agent dynamics represented by equations similar to
\eqref{compensator2b}-\eqref{compensator2c}
 was studied in
\cite{yan2020,hu2020}, called autonomous synchronization, for a leaderless network.
The technical challenge was caused by a specified structure of \eqref{compensator2b} while a free design of observer structure
 is allowed in this paper.  However, the  technical challenge of this paper, compared with \cite{yan2020,hu2020}, is that only $\hat{\eta}_i$, instead of $\hat{\omega}_i$, is allowed to transmit among the agents.
\end{rem}

Let $D=\textnormal{block diag}\left(d_1,\cdots,d_N\right)$,
$e_{vi}=\sum_{j \in \mathcal{\bar{N}}_i}a_{ij} (\hat{\eta}_j-\hat{\eta}_i)$,  $e_{v}=\col(e_{v1},\cdots,e_{vN})$, $\bar{v}=\mathds{1}_N\otimes v$, $\bar{\omega}=\mathds{1}_N \otimes \omega$, $\hat{\omega}=\col\left(\hat{\omega}_1,\cdots,\hat{\omega}_N \right)$, $\hat{\eta}=\col\left(\hat{\eta}_1,\cdots,\hat{\eta}_N\right)$.
Then, we have
\begin{align*}
\dot{\hat{\eta}}&=S_d\left(\hat{\omega}\right)\hat{\eta} +\mu_1\left(D\otimes I_m\right)e_{v},\\
\dot{\hat{\omega}}&=\mu_2\phi_d(e_{v})\left(D\otimes I_m\right)\hat{\eta},
\end{align*}
where
\begin{align}\phi_d\left(e_v\right)&=\textnormal{block diag} \left(\phi(e_{v1}),\cdots,\phi(e_{vN})\right),\nonumber\\
S_d\left(\hat{\omega}\right)&=\textnormal{block diag} \left(S\left(\hat{\omega}_1\right),\cdots,S\left(\hat{\omega}_N\right)\right).\nonumber
\end{align}
By definitions of $e_v$, $H$, $\hat{\eta}$ and $\bar{v}$ and by noting $\hat{\eta}_0= v$, we can show that
\begin{equation} \label{eqev}
e_v =-(H\otimes I_m)(\hat{\eta}-\bar{v}).
\end{equation}
Then, by using the above notations, we have the following compact form
 \begin{subequations}\label{noleader}\begin{align}
\dot{\hat{\eta}}&=S_d\left(\hat{\omega}\right)\hat{\eta} - \mu_1\left(DH\otimes I_m\right)(\hat{\eta}-\bar{v}),\label{noleader1}\\
\dot{\hat{\omega}}&=\mu_2\phi_d(e_{v})\left(D\otimes I_m\right)\hat{\eta}.\label{noleaderd3}
\end{align}\end{subequations}
Define $\tilde{{\eta}}=[\tilde{\eta}_i]=\hat{\eta} - \bar{v}$ and $\tilde{\omega}=[\tilde{\omega}_i]= \hat{\omega} - \bar{\omega}$.  We have
\begin{subequations}\label{leadererror}\begin{align}
\dot{\tilde{\eta}} &=\left(I_N\otimes S(\omega)- \mu_1\left(DH\otimes I_m\right)\right)\tilde{\eta} +S_d\left(\tilde{\omega}\right)\hat{\eta},\label{leadererror1}\\
\dot{\tilde{\omega}}&=\mu_2\phi_d(e_{v})\left(D\otimes I_m\right)\hat{\eta}. \label{leadererror2}
\end{align}
\end{subequations}
Using Lemma 3.3 in \cite{WangHuang2019IJC-coop}, \eqref{leadererror} can be further put into the following form
\begin{subequations}\label{creeq2}
\begin{align}
  \dot{\tilde{\eta}} &=\left(I_N\otimes S\left(\omega\right)-\mu_1DH\otimes I_{m}\right)\tilde{\eta} - \phi_d^T (\hat{\eta})\tilde{\omega},\label{creeq2b}\\
\dot{\tilde{\omega}}&=\mu_2\phi_d\left(\hat{\eta}\right)\left(DH\otimes I_{m}\right)\tilde{\eta},\label{creeq2abbbb}
\end{align}
\end{subequations}
where  $\phi_d\left(\hat{\eta}\right)=\textnormal{block diag} \left(\phi(\hat{\eta}_{1}),\cdots,\phi(\hat{\eta}_{N})\right)$.

Also, the dynamics of the estimation error of the output matrix $\tilde{E}_i$ is
\begin{align}\label{output2}
\dot{\tilde{E}}_i=&d_i\sum\nolimits_{j\in \mathcal{\bar{N}}_i} a_{ij}(\tilde{E}_j-\tilde{E}_i)vv^{T}+\psi_i(t),
\end{align}
where
\begin{align*}
\psi_i(t)=&d_i\sum\nolimits_{j\in \mathcal{\bar{N}}_i} a_{ij}\big((\tilde{E}_j-\tilde{E}_i)v\tilde{\eta}_i^{T}\big)+d_iEe_{vi}\hat{\eta}_i^{T}\nonumber\\
&+d_i\sum\nolimits_{j\in \mathcal{\bar{N}}_i} a_{ij}\big(\tilde{E}_j\tilde{\eta}_j\hat{\eta}_i^{T}-\tilde{E}_i\tilde{\eta}_i\hat{\eta}_i^{T}\big).
\end{align*}
Define $\zeta_i=\mbox{vec}\big(\tilde{E}_i\big)$, $\zeta_0=\mbox{vec}\big(E\big)$ and $\pi_i=\mbox{vec}\big(\psi_i(t)\big)$. Equation \eqref{output2} can be written as
\begin{align}\label{output3}
\dot{\zeta}_i=&\left(vv^{T}\otimes I_n\right)d_i\sum\nolimits_{j\in \mathcal{\bar{N}}_i}a_{ij}\left(\zeta_j-\zeta_i\right)+\pi_i,
\end{align}
where
\begin{align}\label{output2pi}
\pi_i(t)=&d_i\sum\limits_{j\in \mathcal{\bar{N}}_i} a_{ij}\left[\left(\tilde{\eta}_i v^{T}\otimes I_n\right)\big(\zeta_j-\zeta_i)\big)+\left(\hat{\eta}_i\tilde{\eta}_j^{T}\otimes I_n\right)\zeta_j\right.\nonumber\\
&\left.-\left(\hat{\eta}_i\tilde{\eta}_i^{T}\otimes I_m\right)\zeta_i\right]+d_i\left(\hat{\eta}_ie_{vi}^{T}\otimes I_n\right)\zeta_0.
\end{align}
Further, let $\zeta=\col\left(\zeta_1,\cdots,\zeta_N\right)$ and $\pi=\col\left(\pi_1,\cdots,\pi_N\right)$. Then, \eqref{output3} can be put as
\begin{align}\label{outputcom}
\dot{\zeta}=&-\left(I_N\otimes\left(vv^{T}\otimes I_n\right)\right)\left(DH\otimes I_{nm}\right)\zeta+\pi\nonumber\\
=&-\left(DH\otimes\left(vv^{T}\otimes I_n\right)\right)\zeta+\pi.
\end{align}

The observer \eqref{compensator2} is equivalent to the system composed of \eqref{creeq2} and  \eqref{outputcom}
through the above analysis based on coordinate transformation. Therefore, the convergence of the former
can be verified by investigating the stability of the latter, which is the major  objective in Subsection~\ref{Convergence-as}.

\subsection{Some technical lemmas}

Before analysing the stability of systems  \eqref{creeq2} and  \eqref{outputcom},  we present several technical lemmas.
Let us begin with Lemma~\ref{Lm-WP}, which guides the choice of the design parameter $d_i$.

\begin{lem}\label{Lm-WP}
Under Assumption \ref{Ass-graph}, there exists a positive diagonal matrix $D \in \mathds{R}^{N \times N}$ such that
all the eigenvalues of $DH$ are  real, positive, and distinct. Moreover, there   exists a symmetric positive definite matrix $W \in \mathds{R}^{N \times N}$ such that
\begin{equation} \label{Eq-Q}
P=WDH\; \text{and}\;
Q=PDH+H^TDP
\end{equation} are symmetric positive definite.
\end{lem}
\begin{proof}
Under Assumption \ref{Ass-graph}, $H$ is nonsingular, with all its eigenvalues having positive real parts, and all its off-diagonal entries being  nonpositive \cite{HuHong2007PA-leader}. Thus, $H$ is a nonsingular $M$-matrix, and all its leading principal minors are positive.
 By Lemma \ref{dstable}, there exists a positive diagonal matrix $D=\textnormal{block diag}(d_1,\cdots,d_N)$ such that all the eigenvalues of $DH$ are real, positive, and distinct. This implies that $DH$ is diagonalizable. Therefore, by Lemma \ref{matWC},  there exists a symmetric positive definite matrix $W \in \mathds{R}^{N \times N}$ such that $WDH$ is symmetric positive definite. The positive definiteness of $Q$ can be easily shown by observing
 $Q=2H^TDWDH.$
\end{proof}

\begin{rem}
For a general directed communication graph $\mathcal{G}$, the details on how to choose $D=\textnormal{block diag}(d_1,\dots,d_N)$ and $W$ can be found in \cite{Fisherfuller1958} and \cite{Barkana2006Auto-mitigation}, respectively. When  $\mathcal{G}$ is undirected, under Assumption \ref{Ass-graph}, we can pick $D=W=I$. When $\mathcal{G}$ is detailed balanced, under Assumption \ref{Ass-graph}, we can design $W=I$ and $D=\diag(d)$ with $d$ being
the left eigenvector of the Laplacian matrix $L$ corresponding to its zero eigenvalue.
\end{rem}

\begin{lem} \label{Lm-BDH}
Consider $H$ and $D$ as in Lemma \ref{Lm-WP}. There exists a positive diagonal matrix $B=\textnormal{block diag}(b_1,\cdots,b_N)$ such that
\begin{equation}\label{Eq-BDH}
\bar{H}=BDH+H^TDB
\end{equation} is positive definite.
\end{lem}
\begin{proof}
Since $H$ is an $M$-matrix and $D$ is a positive diagonal matrix, $DH$ is an $M$-matrix. Theorem 2.5.3 of \cite{merino1992topics} completes the proof.
\end{proof}
\begin{rem}
Lemmas \ref{Lm-WP} and \ref{Lm-BDH}  play a central role in the convergence analysis of observers. This will be shown in Lemmas \ref{Lm-unif-bund-eta} and \ref{etape}, and Theorem \ref{Thm-1}.
More importantly, equations \eqref{Eq-Q} and \eqref{Eq-BDH} provide two choices of the so-called graph related Lyapunov equations (cf. the Lyapunov equation \cite[Equation 4.12]{khalil2002nonlinear}), which can be used as building blocks for Lyapunov functions in the controller design and stability analysis of MASs over directed graphs. It is well known that analysis of MASs over undirected graphs relies heavily on {symmetry of certain matrices, such as Laplacian matrices}, reflecting the symmetric topology of undirected graphs, while the non-symmetric property of directed graphs constitutes the major challenge for analysis of MASs over directed graphs. The {merits} of equations \eqref{Eq-Q} and \eqref{Eq-BDH} are that they construct two symmetric matrices from the non-symmetric matrix $H$ of directed graphs. Therefore, by using these Lyapunov equations, many existing results of  MASs can be extended from undirected graphs or balanced directed graphs to general directed graphs
(see Example \ref{20210324 Example}).
This further enriches the Lyapunov function design tools in \cite{zhqzl2015} for MASs.
\end{rem}

The following example shows how to use Lemma \ref{Lm-WP} to extend results of \cite{HengsterLewis2014TAC-optimal} to general directed graphs. Due to space limitation, we only provide some key ideas of this extension, instead of detailed derivation.
\begin{exam}\label{20210324 Example}
The work \cite{HengsterLewis2014TAC-optimal} considers the cooperative optimal tracking problem of linear MASs with follower nodes
\begin{align}\label{ex1follower}
\dot{x}_i=A_ox_i+B_o u_i, \quad i=1,2,\dots,N,
\end{align}
and a leader node
\begin{align}\label{ex1leader}
\dot{x}_0=A_ox_0,
\end{align}
where $x_i\in \mathds{R}^{n}$. Its distributed control law in \cite{HengsterLewis2014TAC-optimal} is designed as
\begin{align}\label{20210325 cont1}
u_i=c K \sum\nolimits_{j\in \mathcal{\bar{N}}_i}a_{ij}(x_j-x_i).
\end{align}
The convergence of the tracking error is based on the condition that $H=L+G$ in  \cite[equation (20)]{HengsterLewis2014TAC-optimal} should be real, positive, and distinct. However, this condition does not hold for general directed graphs. Using Lemma \ref{Lm-WP}, we can construct a matrix $D=\textnormal{block diag}(d_1,\cdots,d_N)$. Then, we modify \eqref{20210325 cont1} as
\begin{align}\label{20210325 cont2}
u_i=c d_i K \sum\nolimits_{j\in \mathcal{\bar{N}}_i}a_{ij}(x_j-x_i),
\end{align}
where $K=R_{o}^{-1}B^T_oP_1$ and $P_1$ is positive definite matrix satisfying
\begin{align}
&A^T_oP_1+P_1A_o+Q_o-P_1B_oR_o^{-1}B^T_oP_1=0,\nonumber
\end{align}
 with $Q_o$ and $R_o$ being some positive definite matrices and $$c>\frac{\sigma_{\max}\left(WDH\otimes (Q_o-K^T R_{o}K)\right)}{\sigma_{\min}\left((H^TDWDH)\otimes(K^TR_oK)\right)},$$
 where $W$ is defined in Lemma \ref{Lm-WP}.
 Let $x=\textnormal{\col}(x_1,\cdots,x_N)$ and $\tilde{x}=x-\mathds{1}\otimes x_0$. Then, the closed loop system composed of \eqref{ex1follower}, \eqref{ex1leader} and \eqref{20210325 cont2}, can be put in the following compact form
\begin{align}
\dot{\tilde{x}}=(I_N\otimes A_o-cDH\otimes B_oK)\tilde{x}.\nonumber
\end{align}
Let $P_2=cWDH$ and pick a Lyapunov function candidate as $V(\tilde{x})=\tilde{x}^T\left(P_2\otimes P_1\right)\tilde{x}$. Then, after some mathematical manipulation and applying Lyapunov stability theory, it can be shown that $\lim\limits_{t\rightarrow\infty}\tilde{x}(t)= 0$. Thus the results in \cite{HengsterLewis2014TAC-optimal} is extended to general directed graphs satisfying Assumption \ref{Ass-graph}.
\end{exam}

The following two lemmas show the boundedness of some signals $\hat{\eta}_i$, $\dot{\hat{\eta}}_i$ and $\hat{\omega}_i$, which will be used for the convergence analysis of the observer.
\begin{lem}\label{Lm-unif-bund-eta}
Consider systems \eqref{leader} and \eqref{compensator2}.
 Given two compact sets $\mathds{V}_0\subset \mathds{R}^{m}$ and $\mathds{W}\subset \mathds{R}^{l}$ with $0\in \mathds{V}_0$ and $0\in \mathds{W}$. Under Assumptions \ref{Ass-S} and \ref{Ass-graph}, for any $\hat{\omega}_i(0)\in \mathds{W}$, $\omega\in \mathds{W}$, $\hat{\eta}_i(0)\in \mathds{V}_0$, $v(0)\in \mathds{V}_0$, $\mu_1>1$, and $\mu_2>0$,
 \begin{enumerate}
 \item $\hat{\eta}_i(t)$ is uniformly bounded for $t\geq 0$; and there exists a $t_0$ such that ${\hat{\eta}}_i(t)$ is uniformly bounded, regardless of $\mu_1$ and $t\geq t_0$;
  \item $\dot{\hat{\eta}}_i(t)$ is uniformly bounded for $t\geq 0$; and there exists a $t_1$ such that $\dot{\hat{\eta}}_i(t)$ is uniformly bounded, regardless of $\mu_1$ and $t\geq t_1$; 
  \item $ S({\hat{\omega}}_i(t))\hat{\eta}_i(t)$ is uniformly bounded, regardless of $\mu_1$ and $t$.
\end{enumerate}
\end{lem}
\begin{proof}
See Appendix \ref{Appendix-A}.
\end{proof}
In the following lemma, we will show that $\phi(\hat{\eta}_i(t))$ is persistently exciting for a sufficiently large  $\mu_1$, we first recall the definition of persistently exciting.
\begin{defi}\label{defipe}\cite{Anderson1977}
A uniformly bounded piecewise continuous function $f : [0, +\infty) \mapsto \mathds{R}^{m \times n}$ is said to be persistently exciting (PE) in $\mathds{R}^{m}$  with a level of excitation $\alpha_0$ if there exist positive constants $\alpha_1$, $T_0$ and $T$, such that,
$$\alpha_1 I_m\geq \frac{1}{T}\int^{t+T}_{t} f(s) f^T (s) ds\geq\alpha_0 I_m,~~~~\forall t\geq T_0.$$
\end{defi}
The properties and various other equivalent
definitions of persistently exciting are given in \cite{narendra1987persistent,Pesun1996,adetola2014adaptive}.
\begin{lem}\label{etape}
Consider systems (\ref{leader}) and (\ref{compensator2}). Given two compact sets $\mathds{V}_0\subset \mathds{R}^{m}$ and $\mathds{W}\subset \mathds{R}^{l}$ with $0\in \mathds{V}_0$ and $0\in \mathds{W}$. Under Assumptions \ref{Ass-S} and \ref{Ass-graph}, for any $\hat{\omega}_i(0)\in \mathds{W}$, $\omega\in \mathds{W}$, $\hat{\eta}_i(0)\in \mathds{V}_0$, $v(0)\in \mathds{V}_0$
with $\textnormal{\col} (v_{2k-1} (0), v_{2k} (0)) \neq 0$, $k = 1, \cdots, l$, and $\mu_2>0$, there exists a sufficiently large $\mu_1$ such that $\phi(\hat{\eta}_i(t))$ is PE, and both $\hat{\omega}_i(t)$ and $\tilde{\omega}_i(t)$ are uniformly bounded, regardless of $\mu_1$ and $t$.
\end{lem}
\begin{proof}
See Appendix \ref{Appendix-B}.
\end{proof}

To facilitate the exponential convergence analysis of our adaptive distributed observer, a technical lemma is first presented, which studies the stability of a different class of adaptive systems (c.f. system (B.29) in \cite{marino1996nonlinear}), and is interesting in itself, since it provides a main tool in the analysis of parameter convergence for adaptive observers.

\begin{lem}\label{Lm-exp}
Consider the following system
\begin{align}
  \dot{x} =&\left(I_n\otimes A_a- Y\otimes I_m\right)x- \left(I_n\otimes\psi^T (t)\right)z,\nonumber\\
\dot{z}=&\kappa \left(Y\otimes (\psi(t)P_a)\right)x,\label{creeq3exabbbb}
\end{align}
where $A_a\in  \mathds{R}^{m\times m}$ is such that $P_a A_a +A_a^T P_a \leq 0$ with $P_a\in  \mathds{R}^{m\times m}$ is a positive symmetric definite matrix, $x(t)\in \mathds{R}^{nm}$, $z(t)\in \mathds{R}^{ns}$, $\psi(t)\in \mathds{R}^{s\times m}$, $\kappa$ is a positive constant, and $Y\in \mathds{R}^{n\times n}$ is diagonalizable with positive eigenvalues.
If $\| \psi(t)\|$ and $\| \dot{\psi}(t)\|$ are uniformly bounded, and $\psi(t)$ is PE, then the equilibrium point at the origin is  globally exponentially stable.
\end{lem}
\begin{proof}
See Appendix \ref{Appendix-C}.
\end{proof}
\subsection{Convergence analysis of the observer} \label{Convergence-as}

Now, we are ready to present the first main result about  asymptotic convergence of the observers \eqref{compensator2b} and \eqref{compensator2c}.

\begin{thm}\label{Thm-1}
Consider systems \eqref{leader}, \eqref{compensator2b} and \eqref{compensator2c}. Given two compact sets $\mathds{V}_0\subset \mathds{R}^{m}$ and $\mathds{W}\subset \mathds{R}^{l}$ with $0\in \mathds{V}_0$ and $0\in \mathds{W}$. Under Assumptions \ref{Ass-S} and \ref{Ass-graph}, for any $\hat{\omega}_i(0)\in \mathds{W}$, $\omega\in \mathds{W}$, $\hat{\eta}_i(0)\in \mathds{V}_0$, $v(0)\in \mathds{V}_0$ with
$\textnormal{\col}(v_{2k-1} (0), v_{2k} (0)) \neq 0$,
$k = 1, \cdots, l$,  and $\mu_2>0$, there exists a sufficiently large $\mu_1$ such that $\hat{\eta}(t)$ and $\hat{\omega}(t)$ are uniformly bounded, and
\EQ
&&\lim\limits_{t\rightarrow\infty} \tilde{\eta} (t) =0,  \label{eq101a}\\
&&\lim\limits_{t\rightarrow\infty}\dot{\tilde{\omega}}(t)=0,  \label{eq101b}\\
&&\lim\limits_{t\rightarrow\infty}\tilde{\omega}(t)=0.  \label{eq101c}
 \EN
\end{thm}
\begin{proof}
Pick the following Lyapunov function candidate
\begin{equation}\label{reeq3b}
V =\tilde{\eta}^T\left(P\otimes I_m\right)\tilde{\eta}+\mu_2^{-1}\tilde{\omega}^T\left(W\otimes I_l\right)\tilde{\omega},
\end{equation}
where $P \in \mathds{R}^{N \times N}$ and $W\in \mathds{R}^{N \times N}$ are symmetric positive definite matrices designed in Lemma \ref{Lm-WP}.
The time derivative of \eqref{reeq3b} along \eqref{creeq2} is
\begin{align}
\dot{V}
=&2\tilde{\eta}^T\left(P\otimes S\left(\omega\right)\right)\tilde{\eta} -\mu_1\tilde{\eta}^T\left((PDH+H^TDP)\otimes I_m\right)\tilde{\eta} \nonumber\\
&-2\tilde{\eta}^T\left(P\otimes I_m\right)\phi_d^T (\tilde{\eta})\tilde{\omega}-2\tilde{\eta}^T\left(P\otimes I_m\right)\phi_d^T (\bar{v})\tilde{\omega}\nonumber\\
&+2\mu_2^{-1}\tilde{\omega}^T\left(W\otimes I_l\right)\dot{\tilde{\omega}}.\nonumber
\end{align}
Since $S(\omega)$ is skew symmetric and $P$ is symmetric, $P\otimes S\left(\omega\right)$ is also skew symmetric. Considering $\phi_d (\bar{v})=I_N\otimes\phi(v)$ and equation \eqref{Eq-Q}, under Assumptions \ref{Ass-S} and \ref{Ass-graph}, we have
\begin{align}\label{reeq4bb}
\dot{V}=& -\mu_1\tilde{\eta}^T\left(Q\otimes I_m\right)\tilde{\eta}-2\tilde{\eta}^T\left(P\otimes I_m\right)\phi_d^T (\tilde{\eta})\tilde{\omega}\nonumber\\
&-2\tilde{\eta}^T\left(P\otimes \phi^T(v)\right)\tilde{\omega}+2\mu_2^{-1}\tilde{\omega}^T\left(W\otimes I_l\right)\dot{\tilde{\omega}}.
\end{align}
Substituting \eqref{creeq2abbbb} into \eqref{reeq4bb} leads to
\begin{align}
\dot{V} =& -\mu_1\tilde{\eta}^T\left(Q\otimes I_m\right)\tilde{\eta}\nonumber\\
&-2\tilde{\eta}^T\left(P\otimes I_m\right)\phi_d^T (\tilde{\eta})\tilde{\omega}-2\tilde{\eta}^T\left(P\otimes \phi^T(v)\right)\tilde{\omega}\nonumber\\
&+2\tilde{\omega}^T\left(W\otimes I_l\right)\phi_d\left(\hat{\eta}\right)\left(DH\otimes I_{m}\right)\tilde{\eta}\nonumber\\
=& -\mu_1\tilde{\eta}^T\left(Q\otimes I_m\right)\tilde{\eta}-2\tilde{\eta}^T\left(P\otimes I_m\right)\phi_d^T (\tilde{\eta})\tilde{\omega}\nonumber\\
&+2\tilde{\omega}^T\left(W\otimes I_l\right)\phi_d\left(\tilde{\eta}\right)\left(DH\otimes I_{m}\right)\tilde{\eta}\nonumber\\
&+2\tilde{\omega}^T\left(WDH\otimes \phi(v)\right)\tilde{\eta}-2\tilde{\eta}^T\left(P\otimes \phi^T (v)\right)\tilde{\omega}.\label{reeq4bbbv}
\end{align}
Under Assumption \ref{Ass-graph}, using $P=WDH$ gives
\begin{align}\tilde{\omega}^T\left(WDH\otimes \phi(v)\right)\tilde{\eta}&=\tilde{\omega}^T\left(P\otimes \phi(v)\right)\tilde{\eta}\nonumber\\
&=\tilde{\eta}^T\left(P\otimes \phi^T (v)\right)\tilde{\omega}.\label{reeq4bbbvi}
\end{align}
Hence, from \eqref{reeq4bbbv} and \eqref{reeq4bbbvi}, we have
\begin{align}
\dot{V} =& -\mu_1\tilde{\eta}^T\left(Q\otimes I_m\right)\tilde{\eta}-2\tilde{\eta}^T\left(P\otimes I_m\right)\phi_d^T (\tilde{\eta})\tilde{\omega}\nonumber\\
&+2\tilde{\omega}^T\left(W\otimes I_l\right)\phi_d\left(\tilde{\eta}\right)\left(DH\otimes I_{m}\right)\tilde{\eta}\nonumber\\
\leq & -\mu_1\tilde{\eta}^T\left(Q\otimes I_m\right)\tilde{\eta}\nonumber\\
&+2\left\|\tilde{\omega}\right\|\left\|\phi_d (\tilde{\eta})\right\|\left\|\left(P\otimes I_m\right)\right\|\left\|\tilde{\eta}\right\|\nonumber\\
&+2\left\|\tilde{\omega}\right\|\left\|\left(W\otimes I_l\right)\right\|\left\|\phi_d\left(\tilde{\eta}\right)\right\|\left\|\left(DH\otimes I_{m}\right)\right\|\left\|\tilde{\eta}\right\|.\nonumber
\end{align}
By Proposition 9.4.11 in \cite{Bernstein2009Book-Matrix},
$$\| \phi_d (\tilde{\eta}) \| \leq \| \phi_d (\tilde{\eta}) \|_F = \| \hbox{vec}(\phi_d (\tilde{\eta}))\| = \|\tilde{\eta}\|.$$
Thus,
\begin{align}
\dot{V} \leq &  -\mu_1\tilde{\eta}^T\left(Q\otimes I_m\right)\tilde{\eta}+2\left\|\tilde{\omega}\right\|\left\|P\right\|\left\|\tilde{\eta}\right\|^2\nonumber\\
&+2\left\|\tilde{\omega}\right\|\left\|W\right\|\left\|DH\right\|\left\|\tilde{\eta}\right\|^2. \nonumber
\end{align}
According to Lemma \ref{Lm-WP}, $Q$ is symmetric positive definite. Let $\lambda_q$ be the smallest eigenvalue of $Q$ and
\begin{align}\label{mstar}m^*=2\left\|P\right\|+2\left\|W\right\|\left\|DH\right\|.\end{align}
 Then,
\begin{align}\label{reeq5bbb}
\dot{V}\leq&-\left(\mu_1\lambda_{q}-m^*\|\tilde{\omega}(t)\|\right)\|\tilde{\eta}\|^2.
\end{align}
By Lemma \ref{etape}, $\tilde{\omega}(t)$ is uniformly bounded for $t\geq0$ regardless of $\mu_1$, i.e.,
\begin{align}\label{omegastar}\omega^{*}=\sup\limits_{\begin{array}{c}
                              t\geq0,~\omega\in \mathds{W} \\
                             \omega_i(0)\in \mathds{W},i\in \mathcal{V}
                            \end{array}
}\left\{\|\tilde{\omega}(t)\|  \right\},\end{align}
for some positive number $\omega^{*}$, regardless of $\mu_1$ and $t$.
 Let \begin{align}\label{mubound}\mu_{\max}=\max\left\{\frac{\omega^{*}m^*+1}{\lambda_q},  \mu_1^{*}\right\},\end{align}
 where $\mu_1^{*}$ is given in Lemma \ref{etape}.
 Then, for any $\mu_1\geq \mu_{\max}$,
\begin{align}\label{reeq5c}
\dot{V}\leq&-\|\tilde{\eta}\|^2\leq 0.
\end{align}
Since $\dot{V}$ is negative semi-definite and $V$ is positive definite and lower uniformly bounded,  $\tilde{\eta}$ and $\tilde{\omega}$ are uniformly bounded. From \eqref{creeq2b}, $\dot{\tilde{\eta}}$ is uniformly bounded, which implies that $\tilde{\eta}$ is uniformly continuous. From equation \eqref{reeq5c}, we have
\begin{align}
\int_{0}^{\infty}\|\tilde{\eta}(\tau)\|^2d\tau \leq V(0)-V(\infty).\nonumber
\end{align}
 By Barbalat's lemma \cite{khalil2002nonlinear}, we have
 $$\lim\limits_{t\rightarrow\infty}\tilde{\eta}(t)=0.$$
Combining \eqref{eq101a}  and \eqref{creeq2abbbb} yields \eqref{eq101b}.

What is left is to show the convergence of $\tilde{\omega}$, i.e., \eqref{eq101c}.  Differentiating both sides of \eqref{creeq2b} yields
\begin{align}
 \ddot{\tilde{\eta}}  = & \left(I_N\otimes S\left(\omega\right)-\mu_1DH\otimes I_m\right) \dot{\tilde{\eta}}   +\phi_d^T (\hat{\eta})\dot{\tilde{\omega}} + {\phi}_d^T(\dot{\hat{\eta}})\tilde{\omega}.\nonumber
 \end{align}
Since ${\tilde{\eta}}$, $\tilde{\omega}$, and $\hat{\eta}$ are all  uniformly bounded (see Lemma \ref{Lm-unif-bund-eta}),
by \eqref{noleader1} and \eqref{creeq2}, it is clear that $\dot{\hat{\eta}}$, $\dot{{\tilde{\eta}}}$ and $\dot{\tilde{\omega}}$ are also uniformly bounded. Thus, $\ddot{\tilde{\eta}}$ is uniformly bounded. Again, by applying Barbalat's lemma, we have $\lim_{t \to \infty} \dot{\tilde{\eta}}(t)=0$, which, together with \eqref{eq101a} and \eqref{creeq2b}, further
 implies
 $$\lim\limits_{t\rightarrow\infty} \phi_d^T (\hat{\eta}(t))\tilde{\omega}(t) =0.$$

By Lemma \ref{etape}, under Assumptions  \ref{Ass-S} and \ref{Ass-graph}, for any $v (0)\in \mathds{V}_0$ satisfying $\col (v_{2k-1} (0), v_{2k} (0)) \neq 0$, $k = 1, \cdots, l$, and any $\mu_2>0$, if  $\mu_1\geq \mu_{\max}$,  then $\phi_d (\hat{\eta})$ is PE.
Also note that \eqref{eq101b} holds and $$ \lim\limits_{t\rightarrow\infty}  \tilde{\omega}^T(t)\phi_d (\hat{\eta}(t))\phi_d^T (\hat{\eta}(t))\tilde{\omega}(t) =0.$$ Then \eqref{eq101c} holds according to Lemma 1 in \cite{zbdelyon2001}. Note that  \eqref{eq101c} can also be shown by \cite[Lemma 4.1]{liu2009parameter} or \cite[Lemma 2.4]{chen2015stabilization}, provided that $\phi_d (\hat{\eta})$ is PE and $\lim_{t \to \infty} \dot{\tilde{\omega}}(t)=0$.
\end{proof}
\begin{rem}
Considering \eqref{mubound} and \eqref{mu1-star-2}, a lower bound of $\mu_1$ is given by
$$\mu_{\max}=\max\left\{\frac{\omega^{*}m^*+1}{\lambda_q},  \frac{32\omega^{*}b_{M}^2N\|v(0)\|^2q_{\eta}}{C_{\min}^2 \lambda_{h} b_{m}}\right\},$$
where $m^{*}$, $\omega^{*}$, $q_{\eta}$ and $C_{\min}$ are given in \eqref{mstar}, \eqref{omegastar},  \eqref{qetabound} and \eqref{betav0}, respectively; $b_m$ and $b_M$ denote the smallest and largest eigenvalues of $B$ given in Lemma \ref{Lm-BDH}; $\lambda_h$ is the smallest eigenvalue of $\bar{H}$ given in \eqref{Eq-BDH}; and $\lambda_q$ is the smallest eigenvalue of $Q$ given in \eqref{Eq-Q}. It is worth noting that the lower bound is sometimes practically conservative. In real applications we do not always select the parameter $\mu_1$ according to this formula. Selection based on trial
and error is more practically efficient, while the lower bound guarantees the existence of such a selection. Such a selection method is well accepted in many engineering applications.
\end{rem}


Theorem \ref{Thm-1} shows that $\tilde{\eta}$ and $\tilde{\omega}$ converge to zero asymptotically. Next, we shall further show that the convergence of $\tilde{\eta}$ and $\tilde{\omega}$ is in fact exponential by using Lemma \ref{Lm-exp}.

\begin{thm}\label{Thm-exp}
Consider systems \eqref{leader}, \eqref{compensator2b} and \eqref{compensator2c}.  Given two compact sets $\mathds{V}_0\subset \mathds{R}^{m}$ and $\mathds{W}\subset \mathds{R}^{l}$ with $0\in \mathds{V}_0$ and $0\in \mathds{W}$. Under Assumptions \ref{Ass-S} and \ref{Ass-graph}, for any $\hat{\omega}_i(0)\in \mathds{W}$, $\omega\in \mathds{W}$, $\hat{\eta}_i(0)\in \mathds{V}_0$, $v(0)\in \mathds{V}_0$ with
$\textnormal{\col}(v_{2k-1} (0), v_{2k} (0)) \neq 0$,
$k = 1, \cdots, l$, and $\mu_2>0$, there exists a sufficiently large $\mu_1$ such that $\hat{\eta}(t)$ and $\hat{\omega}(t)$  are uniformly bounded and satisfy
\begin{align*}
\lim\limits_{t\rightarrow\infty} \tilde{\eta} (t) =0 ~~~~{and} ~~~~ \lim\limits_{t\rightarrow\infty}\tilde{\omega}(t)=0, 
\end{align*}
 exponentially.
\end{thm}
\begin{proof}
Let $A=\left(I_N\otimes S\left(\omega\right)-\mu_1DH\otimes I_{m}\right)$ and $\chi=\col\left(\tilde{\eta},\tilde{\omega}\right)$. Then, system \eqref{creeq2} can be written as
\begin{align}\label{creeq4ex}
\dot{\chi}&=\left(\mathcal{M}(t)+\mathcal{M}_0(t)\right)\chi,
\end{align}
where
\begin{align}
\mathcal{M}(t)&=\left[
       \begin{matrix}
        \left(I_N\otimes S\left(\omega\right)-\mu_1DH\otimes I_{m}\right) & - I_N\otimes \phi^T (v(t))  \\
         \mu_2 DH\otimes \phi (v(t)) &  0 \\
       \end{matrix}
     \right],\nonumber\\\mathcal{M}_0(t)&=\left[
                        \begin{matrix}
                         0&  -\phi_d^T (\tilde{\eta}(t)) \\
                         \mu_2\phi_d\left(\tilde{\eta}(t)\right)(DH\otimes I_{m}) & 0 \\
                        \end{matrix}
                      \right],\nonumber
\end{align}
and $\phi_d\left(\tilde{\eta}(t)\right)=\textnormal{block diag}  \left(\phi(\tilde{\eta}_{1}(t)),\cdots,\phi(\tilde{\eta}_{N}(t))\right)$.

System $\dot{\chi}=\mathcal{M}(t)\chi$ is of the form \eqref{creeq3exabbbb} with $Y=\mu_1DH$, $P_a=I_m$ and $A_a=S(w)$. Under Assumption \ref{Ass-S}, since $\col (v_{2k-1} (0), v_{2k} (0)) \neq 0$,  $\forall k = 1, \cdots, l$, $\phi(v)$ is PE, and so is $\phi_d(\bar{v})$. By Lemma \ref{Lm-exp}, we know that the origin of system $\dot{\chi}=\mathcal{M}(t)\chi$ is exponentially stable. By Theorem \ref{Thm-1}, for a sufficiently large $\mu_1$ we have $\lim_{t\rightarrow\infty}\tilde{\eta}(t)=0$, which implies $\lim_{t\rightarrow\infty}\|\mathcal{M}_0(t)\|=0$. Then, by the Example~9.6 in \cite{khalil2002nonlinear}, we have $\lim\limits_{t\rightarrow\infty}\chi(t)=0$ exponentially.
\end{proof}


Finally, the following result shows the convergence of observer \eqref{compensator2-E}.

\begin{thm}\label{k3}
Consider systems  \eqref{leader} and \eqref{compensator2}. Given two compact sets $\mathds{V}_0\subset \mathds{R}^{m}$ and $\mathds{W}\subset \mathds{R}^{l}$ with $0\in \mathds{V}_0$ and $0\in \mathds{W}$. Under Assumptions \ref{Ass-S} and \ref{Ass-graph}, for any $\hat{\omega}_i(0)\in \mathds{W}$, $\omega\in \mathds{W}$, $\hat{\eta}_i(0)\in \mathds{V}_0$, $v(0)\in \mathds{V}_0$ with $\textnormal{\col} (v_{2k-1} (0), v_{2k} (0)) \neq 0$, $k = 1, \cdots, l$, $\mu_2>0$, and $\Hat{E}_i(0)\in \mathds{R}^{n\times m}$, there exists a sufficiently large $\mu_1$ such that $\Hat{E}_i(t)$ and $\dot{\hat{E}}_i(t)$ are uniformly bounded and
$$\lim\limits_{t\rightarrow\infty}\big(\Hat{E}_i(t)-E\big)=0.$$
\end{thm}
\begin{proof}
Let us first consider system \eqref{outputcom} with $\pi=0$, i.e.,
\begin{align}\label{outputcom1}
\dot{\zeta}=&-\left(DH\otimes\left(v(t)v^{T}(t)\otimes I_n\right)\right) \zeta.
\end{align}
Since $DH$ is diagonalizable and has real and positive eigenvalues, there exist a nonsingular matrix $P_H$ such that
$$P_HDHP^{-1}_H=\textnormal{block diag} (\lambda_1,\cdots,\lambda_N )=J_{H},$$ where $\lambda_i>0 ~ (i=1,\dots,N)$ are eigenvalues of $DH$.
Let $x=\left(P_H\otimes I_{nm}\right)\zeta$, then system \eqref{outputcom1} can be transformed into the following system
\begin{align}\label{outputcom2}
\dot{x}=&-\left(J_{H}\otimes\left(v(t)v^{T}(t)\otimes I_n\right)\right)x.
\end{align}
Under Assumption \ref{Ass-S}, for $\col (v_{2k-1} (0), v_{2k} (0)) \neq 0 ~(\forall k = 1, \cdots, l)$, from Lemma 3 of \cite{wanghang2019TNN}, $v(t)$ is  PE.
Moreover, considering the structure of the matrix $v(t)v^T(t)$, there exist positive constants $\epsilon_1$, $\epsilon_2$, $t_0$, and $T_0$ such that, $\forall t\geq t_0$,
$$\epsilon_1 I_m \geq\int^{t+T_0}_{t}  v(\tau) v^T (\tau) d \tau \geq\epsilon_2 I_m.$$
Hence
$$\epsilon_1 J_{H}\otimes  I_{nm}\geq \int^{t+T_0}_{t}J_{H}\otimes(v(\tau)v^{T}(\tau)\otimes I_n) d\tau \geq \epsilon_2 J_{H}\otimes  I_{nm}.$$
Thus, by Theorem 1 in \cite{Anderson1977}, the equilibrium point of system \eqref{outputcom2} is exponentially stable. Equivalently, the equilibrium point of system \eqref{outputcom1} is exponentially stable. The rest proof follows similar development as in Theorem 4.1 of \cite{WangHuang2020IJC-adapt} and is thus omitted for brevity.
\end{proof}

\section{Distributed controller design} \label{Sect-controller}
From Section \ref{Sect-observer}, we have known that each node can asymptotically observe the leader's information, including $v$, $\omega$ and $E$. Now we are ready to propose the control law for each follower node $i$.
As in \cite{CaiHuang2016TAC-EL}, for $i=1,\cdots,N$, let
\begin{subequations}\label{claw}\begin{align}
\dot{q}_{ri}=&\hat{E}_iS\left(\hat{\omega}_i\right)\hat{\eta}_i-\alpha\big(q_i-\Hat{E}_i\hat{\eta}_i\big), \label{claw1}\\
 s_i=&\dot{q}_i-\dot{q}_{ri},\label{claw2}\\
   \tau_i=&-K_{i}s_i+M_i\left(q_i\right)\ddot{q}_{ri}+C_i\left(q_i,\dot{q}_i\right)\dot{q}_{ri}+G_i\left(q_i\right),\label{dstureq11i}
 \end{align}
\end{subequations}
where $K_i$ is a positive definite matrix, $\alpha$ is a positive number and $\Hat{E}_i$ are generated by \eqref{compensator2}.

\begin{thm}
Consider systems \eqref{leader} and \eqref{compensator2}. Under Assumptions \ref{Ass-S} and \ref{Ass-graph}, for any $\hat{\omega}_i(0)\in \mathds{W}$, $\omega\in \mathds{W}$, $\hat{\eta}_i(0)\in \mathds{V}_0$, $v(0)\in \mathds{V}_0$ with $\textnormal{\col} (v_{2k-1} (0), v_{2k} (0)) \neq 0$, $k = 1, \cdots, l$,    $\mu_2>0$, and $\Hat{E}_i(0)\in \mathds{R}^{n\times m}$, there exists a sufficiently large $\mu_1$ such that  Problem \ref{Problem-track} is solvable by the control law \eqref{claw} along with the observer \eqref{compensator2}.

\end{thm}
\begin{proof}
Differentiating both sides of \eqref{claw1} and \eqref{claw2} gives,
\begin{subequations}\label{dclaw}
\begin{align}
\ddot{q}_{ri}=&\dot{\hat{E}}_iS\left(\hat{\omega}_i\right)\hat{\eta}_i+\hat{E}_iS\left(\hat{\omega}_i\right)\dot{\hat{\eta}}_i\nonumber\\
&+\hat{E}_iS\big(\dot{\hat{\omega}}_i\big)\hat{\eta}_i-\alpha\big(\dot{q}_i-\dot{\hat{E}}_i\hat{\eta}_i-\Hat{E}_i\dot{\hat{\eta}}_i\big),\label{dclaw1}\\
 \dot{s}_i=&\ddot{q}_i-\ddot{q}_{ri}. \label{dclaw2}
 \end{align}
\end{subequations}
Considering  \eqref{Model-follower}, \eqref{dclaw} and \eqref{claw}, gives
\begin{align}
 M_i\left(q_i\right)\dot{s}_i&=-C_i\left(q_i,\dot{q}_i\right)s_i-K_is_i.\label{MARINEVESSEL14}
\end{align}
Define the following Lyapunov function:
\begin{equation}\label{MARINEVESSEL15}
V_i=\frac{1}{2}s_{i}^TM_i\left(q_i\right)s_i,
\end{equation}
Then, along the trajectory of (\ref{MARINEVESSEL14}),
\begin{align} 
  \dot{V}_i=&s_{i}^TM_i\left(q_i\right)\dot{s}_i+\frac{1}{2}s_{i}^T\dot{M}_i\left(q_i\right)s_i\nonumber\\
  =&s_{i}^T\left(-C_i\left(q_i,\dot{q}_i\right)-K_i\right)s_i+\frac{1}{2}s_{i}^T\dot{M}_i\left(q_i\right)s_i.\nonumber
  \end{align}
 Noting that $\big(\dot{M}_i\left(q_i\right)-2C_i\left(q_i,\dot{q}_i\right)\big)$ is skew symmetric, we have
\begin{equation}
  \dot{V}_i=-s_{i}^T K_is_i.\nonumber
  \end{equation}
 Since $K_i$ are positive definite, the vectors $s_{i}$ is uniformly bounded. We now further show $s_i(t)\rightarrow0$ as $t\rightarrow\infty$ using Barbalat's lemma.
For this purpose, substituting (\ref{claw1}) into (\ref{claw2}) gives
\begin{equation}\label{MARINEVESSEL17}
 \dot{q}_i= -\alpha q_i+ s_i+\hat{E}_iS\left(\hat{\omega}_i\right)\hat{\eta}_i+\alpha \hat{E}_i\hat{\eta}_i.
 \end{equation}
 Since $s_i$ is uniformly bounded, $\hat{E}_i$, $S\left(\hat{\omega}_i\right)$ and $\hat{\eta}_i$ are also uniformly bounded by Theorem \ref{Thm-1} and Theorem \ref{k3}, and $ \alpha  $ is positive, equation (\ref{MARINEVESSEL17}) can be viewed as a stable first order linear system in $q_i$ with a uniformly bounded input. Thus, both $q_i$ and $\dot{q}_i$ are uniformly bounded, and so are $\dot{q}_{ri}$ and $\ddot{q}_{ri}$, according to \eqref{claw1}  and \eqref{dclaw1}.
Equation \eqref{dclaw2} further implies that $\dot{s}_i$ is  uniformly bounded. Thus $\ddot{V}_i$ are uniformly bounded, which implies $\dot{V}_i$ is uniformly continuous. By Barbalat's lemma, we have, for $i=1,\cdots,N$, $\dot{V}_i(t)\rightarrow 0$ as $t\rightarrow \infty$, which implies  $s_i(t)\rightarrow 0$ as $t\rightarrow \infty$.

Using \eqref{compensator2b} and \eqref{MARINEVESSEL17}, we have
   \begin{eqnarray}\label{MARINEVESSEL19}
   \dot{q}_{i}-\hat{E}_i\dot{\hat{\eta}}_i+\alpha(q_{i}-\hat{E}_i\hat{\eta}_i) =s_i-\mu_1 d_i \hat{E}_i\sum\limits_{j \in \mathcal{\bar{N}}_i}(\hat{\eta}_j-\hat{\eta}_i).
  \end{eqnarray}
Let $e_i=(q_i-\hat{E}_i\hat{\eta}_i)$, \eqref{MARINEVESSEL19} can be rewritten as
   \begin{equation}\label{reMARINEVESSEL19}
 \dot{e}_{i} = -\alpha e_i + s_i-\mu_1 d_i \hat{E}_i e_{vi}.
  \end{equation}
By  Theorem \ref{Thm-1}, $\lim_{t\rightarrow\infty}e_{vi}(t)= 0$.  Equation \eqref{reMARINEVESSEL19} can be viewed as a stable linear system with
$(s_i-\mu_1 d_i \hat{E}_ie_{vi})$ as the input, which is uniformly bounded over $t\geq 0$ and vanishes at the origin. Thus, we conclude that both $e_i=(q_i-\hat{E}_i\hat{\eta}_i)$ and $\dot{e}_i=(\dot{q}_i-\hat{E}_i\dot{\hat{\eta}}_i)$ are uniformly bounded over $t\geq 0$ and will tend to zero as $t\rightarrow \infty$. These facts together with $\lim\limits_{t\rightarrow\infty}(\hat{E}_i(t)\dot{\hat{\eta}}_i(t)-E\dot{v}(t))=0$ and $\lim\limits_{t\rightarrow\infty}(\hat{E}_i(t)\hat{\eta}_i(t)-E{v}(t))=0$ complete the proof.
\end{proof}

\section{Simulation example} \label{Sect-example}

Consider a MAS consisting of four follower nodes and one leader node, whose communication network is shown in Figure \ref{fig1}. Each follower node describes a two-link robot arm \cite{lewis2004Book-robot} whose motion equation takes the following  form
\begin{equation}\label{numerica1}
  M_i\left(q_i\right)\ddot{q}_i+C_i\left(q_i,\dot{q}_i\right)\dot{q}_i+G_i\left(q_i\right)=\tau_i, ~~i=1,\dots, 4,
\end{equation}
where $q_i=\col\left(q_{i1}, q_{i2}\right)$,
\begin{subequations}
\begin{align}
  M_i\left(q_i\right) &= \left(
                            \begin{matrix}
                              a_{i1} +a_{i2}+2a_{i3}\cos q_{i2} &  a_{i2}+a_{i3}\cos q_{i2} \\
                               a_{i2}+ a_{i3}\cos q_{i2}&  a_{i2} \\
                            \end{matrix}
                          \right),\nonumber \\
  C_i\left(q_i,\dot{q}_i\right) &= \left(
                                      \begin{matrix}
                                        -a_{i3}\left(\sin q_{i2}\right)\dot{q}_{i2}& -a_{i3}\left(\sin q_{i2}\right)\left(\dot{q}_{i1}+\dot{q}_{i2}\right) \\
                                         a_{i3}\left(\sin q_{i2}\right)\dot{q}_{i1} & 0  \\
                                      \end{matrix}
                                    \right),
  \nonumber\\
 G_i\left(q_i\right) &=\left(
                          \begin{matrix}
                            a_{i4}g\cos q_{i1}+a_{i5}g\cos\left( q_{i1}+ q_{i2}\right) \\
                            a_{i5}g\cos\left( q_{i1}+ q_{i2}\right) \\
                          \end{matrix}
                        \right),
 \nonumber
\end{align}
\end{subequations}
and the actual values of $\Theta_i=\col(a_{i1},\cdots,a_{i5})$ are
\begin{eqnarray*}
 && \Theta_1 = \col(0.64,1.10, 0.08,0.64,0.32), \\
  && \Theta_2 =\col(0.76,1.17,0.14,0.93,0.44),\\
  &&\Theta_3 = \col(0.91,1.26,0.22, 1.27,0.58),\\
  &&\Theta_4 = \col(1.10,1.36,0.32,1.67,0.73).
\end{eqnarray*}
 The leader's signal is
generated by \eqref{leader} with $$\omega=\col\left(10,20,30\right)~~\textnormal{and}~~E=\left[
                                                                                 \begin{matrix}
                                                                                   1 & 0 & 2 & 0 & 3 & 0 \\
                                                                                   0 & 3 & 0 & 2 & 0 & 1 \\
                                                                                 \end{matrix}
                                                                               \right]
.$$ Thus, Assumption \ref{Ass-S} is satisfied. The communication network satisfies Assumption \ref{Ass-graph} with
$$H=\left[
      \begin{matrix}
       2   &  0  &   0  &  -1\\
    -1   &  2  &  -1   &  0\\
     0   & -1  &   1   & 0\\
    0  &   -1  &  -1  &  2\\
      \end{matrix}
    \right].
$$
We can pick $D=\textnormal{block diag} (1,2,3,4)$ such that $DH$ is diagonalizable and all of its eigenvalues are real,  positive, and distinct.
Also, we can find $W$ and $P$ as follows:
\begin{align}W&=\left[
      \begin{matrix}
       0.7993 &  -0.4421  & -0.0671 &   0.4092\\
   -0.4421   & 1.1599  &  0.4099   &-0.8280\\
   -0.0671   & 0.4099  &  0.6599  & -0.6405\\
    0.4092   &-0.8280  & -0.6405  &  1.1979\\
           \end{matrix}
    \right], \nonumber \\
P&= \left[
       \begin{matrix}
          2.4828 &  -3.2040  & -0.9540  &  2.4745\\
   -3.2040   & 6.7221   & 2.2221 &  -6.1822\\
   -0.9540  &  2.2221  &  3.7221  & -5.0572\\
    2.4745  & -6.1822  & -5.0572 &   9.1741\\
         \end{matrix}
     \right].\nonumber
\end{align}
It can be verified that $P>0$ and $W>0$.
We design a distributed observer of the form \eqref{compensator2} with  $\mu_1=80$, $\mu_2=60$ and a control law of the form \eqref{claw} with $K_i=0.5I_2$ and $\alpha=0.5$. The initial condition of the leader system is $v(0)=\col(2, 0.6, 2, 0.8, 2, 1)$.
Since $v(0)$ satisfies the condition of Lemma 4 in \cite{WangHuang2019IJC-coop}, both $\phi\left(v (t)\right)$ and $v(t)$ are PE.
Figures \ref{fig2}, \ref{fig6i} and \ref{figEe} show the convergence of the estimation errors $\|\tilde{\eta}_{i}\|$, $\|\tilde{E}_i\|$ and $\|\tilde{\omega}_i\|$, and Figure \ref{figqe} shows the asymptotic tracking performance of the follower nodes.

\begin{figure}[H]
\centering
\begin{tikzpicture}
    \node (home) [circle, draw=blue!80, fill=blue!20, very thick, minimum size=7mm] {\textbf{0}};
    \node (school) [circle, right=of home, draw=red!80, fill=red!20, very thick, minimum size=7mm] {\textbf{1}};
    \draw[very thick,->, right] (home) edge (school);
    \node (school2) [circle, right=of school, draw=red!80, fill=red!20, very thick, minimum size=7mm] {\textbf{2}};
    \draw[very thick,->,  bend  right] (school) edge (school2);
    \node (school3) [circle, right=of school2, draw=red!80, fill=red!20, very thick, minimum size=7mm] {\textbf{3}};
         \draw[very thick,->,  bend  right] (school3) edge (school2);
          \draw[very thick,->,  bend  right] (school2) edge (school3);
    \node (school4) [circle, right=of school3, draw=red!80, fill=red!20, very thick, minimum size=7mm] {\textbf{4}};
    \draw[very thick,->,  bend right] (school3) edge (school4);
    \draw[very thick,->,  bend  right] (school2) edge (school4);
        \draw[very thick,<-,, bend left] (school) edge (school4);
\end{tikzpicture}
\caption{ Communication graph $\bar{\mathcal{G}}$}\label{fig1}
\end{figure}
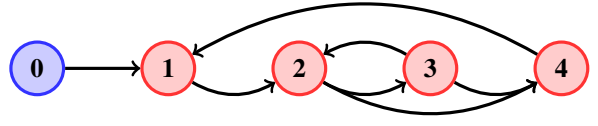

\begin{figure}[ht]
  \centering\setlength{\unitlength}{0.65mm}
  \epsfig{figure=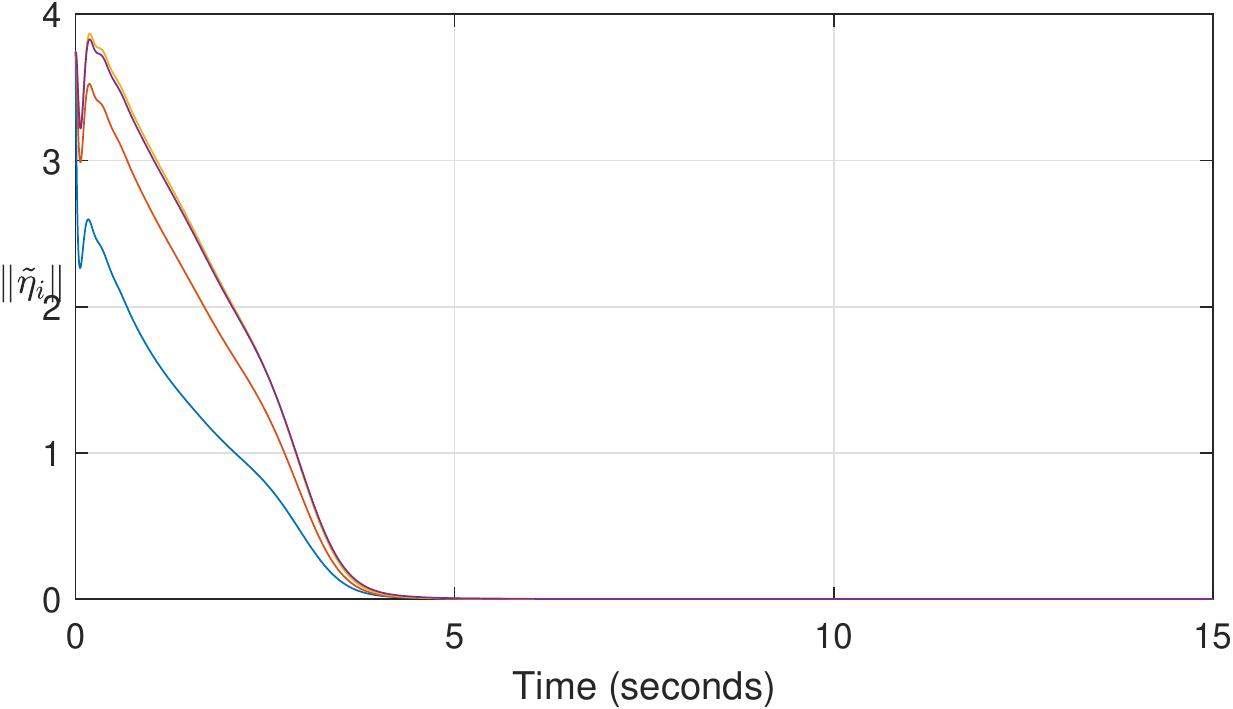,width=8cm}
   \caption{Trajectory of $\|\tilde{\eta}_{i}(t)\|$, $i=1,2,3,4$. }\label{fig2}
\end{figure}

\begin{figure}[ht]
  \centering\setlength{\unitlength}{0.65mm}
  \epsfig{figure=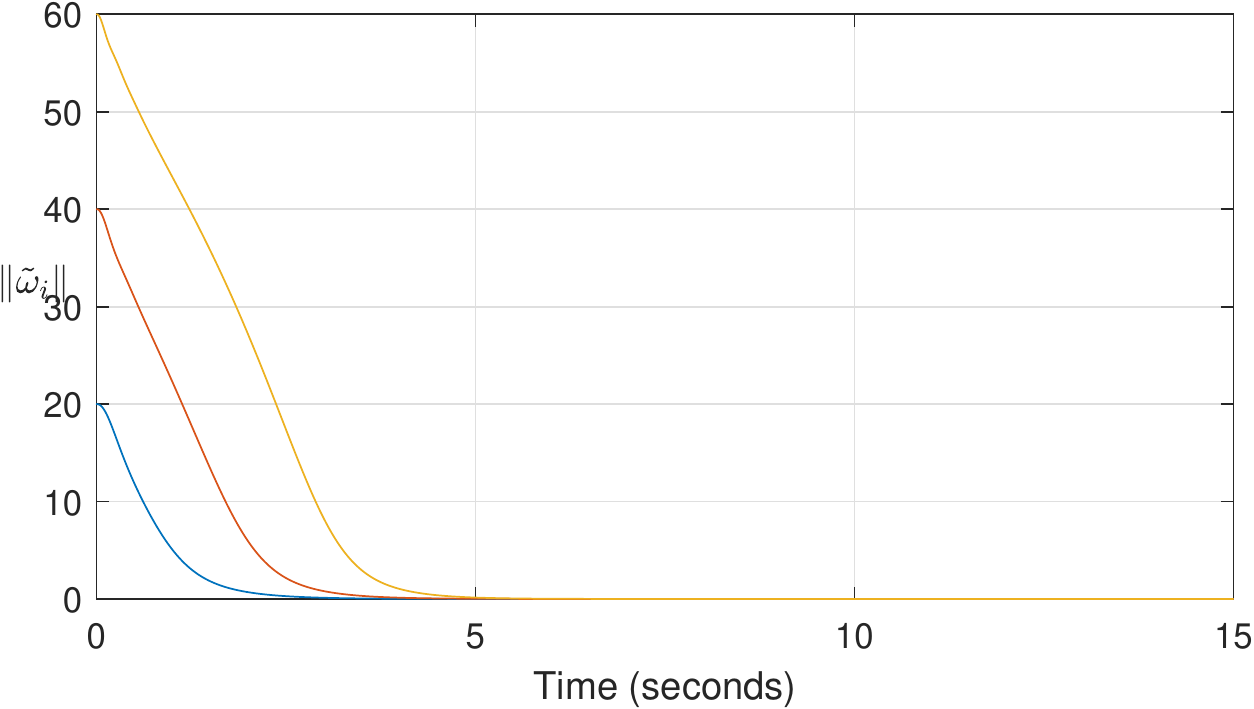,width=8cm}
  \caption{Trajectory of $\|\tilde{\omega}_i(t)\|$, $i=1,2,3,4$.}\label{fig6i}
\end{figure}

\begin{figure}[ht]
  \centering\setlength{\unitlength}{0.65mm}
  \epsfig{figure=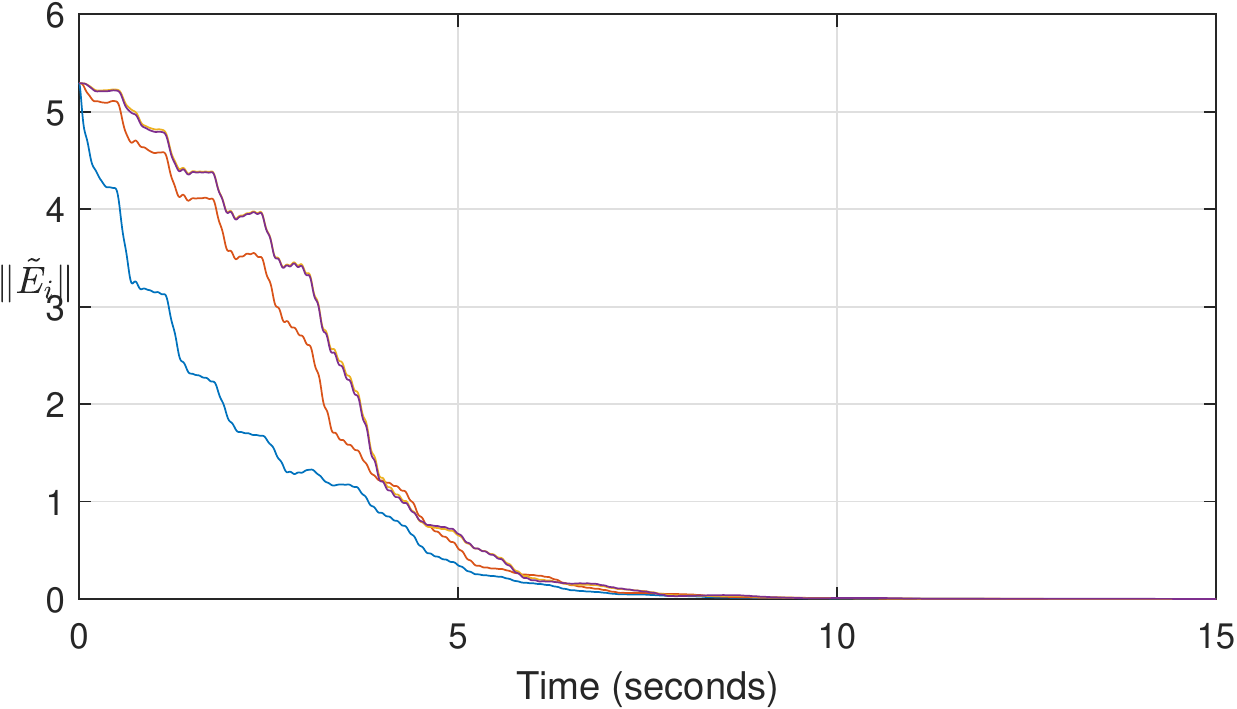,width=8cm}
  \caption{Trajectory of $\|\tilde{E}_i(t)\|$, $i=1,2,3,4$.}\label{figEe}
\end{figure}

\begin{figure}[ht]
  \centering\setlength{\unitlength}{0.65mm}
  \epsfig{figure=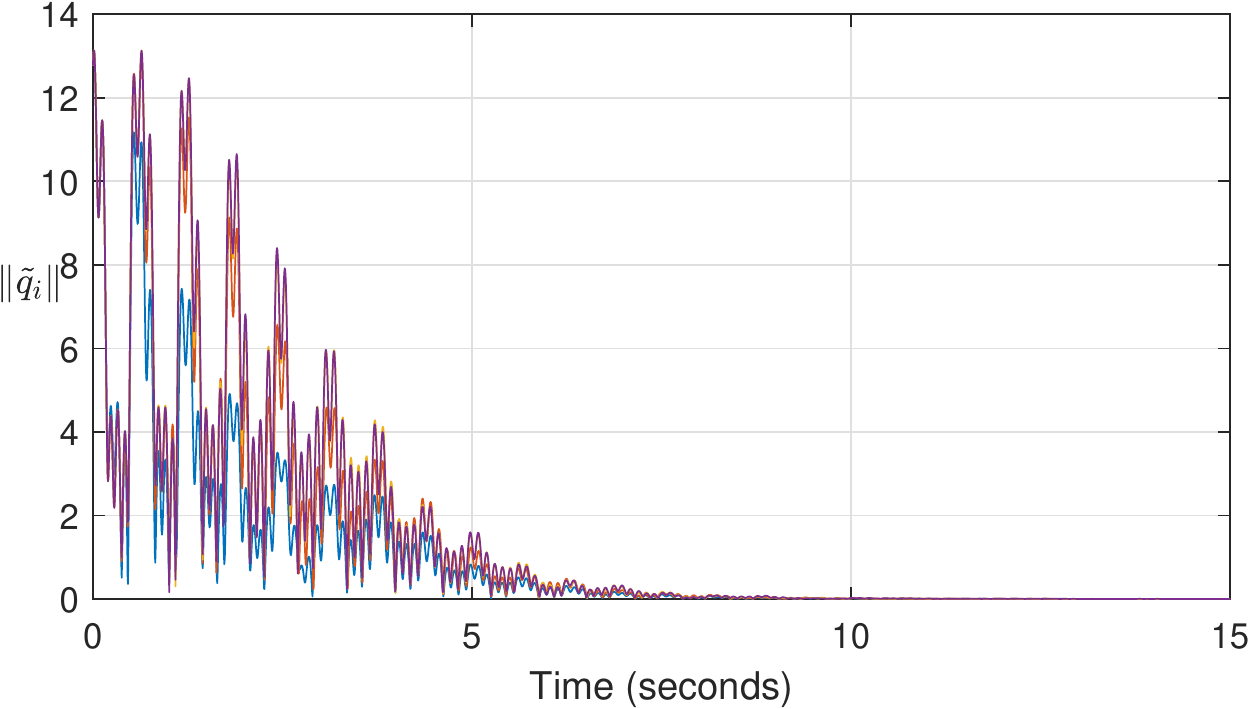,width=8cm}
  \caption{Trajectory of $\|\tilde{q}_i(t)\|=\|q_i(t)-y(t)\|$, $i=1,2,3,4$.}\label{figqe}
\end{figure}

\section{Conclusion} \label{Sect-Con}
This paper studied the leader-following consensus problem of heterogeneous EL systems over directed graphs, where the leader's system matrix and output matrix are both unknown. By suitably picking a diagonal matrix $D$, which solely depends on the communication graph, distributed observers are designed for each follower node, simultaneously estimating both state and parameters of the leader node. The observers are shown to converge exponentially fast. Through the convergence analysis of the distributed observers, two novel graph based Lyapunov equations are developed, which have their own interests in enriching the design and analysis tools for MASs over directed graphs; and the stability of a class of linear time-varying systems is also investigated, which plays a central role in analysis of parameter convergence for adaptive observers. It is also worth mentioning that, due to the separation principle, the proposed observers may also be applied to cooperative tracking problems of other MASs with an uncertain leader, besides EL systems.

 \begin{appendices}
\section{Proof of Lemma \ref{Lm-unif-bund-eta}} \label{Appendix-A}
Define the following Lyapunov function for \eqref{noleader1}
$$X_0(\hat{\eta})=\hat{\eta}^T\left(B\otimes I_m\right)\hat{\eta}, $$
where $B=\textnormal{block diag}(b_1,\dots,b_N)$ is a positive diagonal matrix defined in Lemma \ref{Lm-BDH}. Let $b_m$ and $b_M$ denote the smallest and {\myr largest} eigenvalues of $B$, respectively. Then, the derivative of $X_0(t)$ along the trajectory of \eqref{noleader1} is
\begin{align}
\dot{X}_0
=&2\hat{\eta}^T\left(B\otimes I_m\right)S_d\left(\hat{\omega}\right)\hat{\eta}-2\mu_1\hat{\eta}^T\left(BDH\otimes I_m\right)\hat{\eta}\nonumber\\
&+2\mu_1\hat{\eta}^T\left(BDH\otimes I_m\right)\bar{v}\nonumber\\
= &2\mu_1\hat{\eta}^T\left(BDH\otimes I_m\right)\bar{v}- \mu_1\hat{\eta}^T\left(\bar{H}\otimes I_m\right)\hat{\eta}\nonumber\\
&+\sum\nolimits_{i=1}^{N}2b_i\hat{\eta}_i^TS\left(\hat{\omega}_i\right)\hat{\eta}_i. \nonumber
 \end{align}
Since $S(\hat{\omega}_i)$ is skew symmetric, by using $X_0(\hat{\eta})\leq  b_M \|\hat{\eta}\|^2$, we have
\begin{align}
\dot{X}_0=&- \mu_1\hat{\eta}^T\left(\bar{H}\otimes I_m\right)\hat{\eta}+2\mu_1\hat{\eta}^T\left(BDH\otimes I_m\right)\bar{v}\nonumber\\
\leq&-\mu_1\lambda_{h}\|\hat{\eta}\|^2+2\mu_1 \|\bar{v}\|\|BDH\|\|\hat{\eta}\|\nonumber\\
\leq&-\frac{3\mu_1\lambda_{h}}{4}\|\hat{\eta}\|^2+\mu_1\frac{4}{\lambda_{h}}N\|v\|^2 \|BDH\|^2\nonumber\\
\leq&-\mu_1\gamma X_0+\mu_1\|v\|^2q_0^*,\nonumber
 \end{align}
where $\lambda_h$ is the smallest eigenvalue of $\bar{H}$, $\gamma=\frac{3\lambda_{h}}{4b_M}$ with $b_M$ the largest eigenvalue of $B$, and $q_0^*=\frac{4}{\lambda_{h}} N\|BDH\|^2$. By the Comparison Lemma \cite[Lemma 3.4]{khalil2002nonlinear},  $X_0(t)$ satisfies the inequality
\begin{align*}
X_0(t)&\leq e^{-\mu_1\gamma t}X_0(0)+\mu_1q_0^*\int_{0}^{t}\|v(\tau)\|^2e^{-(t-\tau)\mu_1\gamma}d\tau.
\end{align*}
For any $v(0)\in \mathds{R}^{m}$,  we have $v(t)=e^{S(\omega)t}v(0)$ and $\dot{v}(t)=S(\omega)e^{S(\omega)t}v(0)$.
Under Assumption \ref{Ass-S}, $S(\omega)$ is skew symmetric, $\|v(t)\|=\|v(0)\|$ and $\|\dot{v}(t)\|\leq\|\omega\|\|v(0)\|$.
Then,
\begin{align}
X_0(t)
&\leq e^{-\mu_1\gamma t}X_0(0)+\frac{q_0^*\|v(0)\|^2}{\gamma}.\nonumber
\end{align}
Thus, by using $ b_m \|\hat{\eta}\|^2\leq X_0(\hat{\eta})$, we have
\begin{align}\|\hat{\eta}(t)\|^2\leq \frac{ b_M}{b_m}\|\hat{\eta}(0)\|^2e^{-\mu_1\gamma t}+\frac{q_0^*\|v(0)\|^2}{\gamma b_m}.\tag{A.1}
\end{align}
Hence, $\hat{\eta}(t)$ is uniformly bounded, regardless of $\mu_1$ and $t$. Therefore, for some $t_0$, we have
\begin{align}\label{qetabound}\|\hat{\eta}(t)\|\leq \eta_M\triangleq \|v(0)\|q_{{\eta}},\;\forall t\geq t_0,\tag{A.2}
\end{align}
where $q_{{\eta}}=4\lambda_{h}\|BDH\|\sqrt{\frac{2N b_M}{3b_m}}$.

Next, differentiating both sides of \eqref{noleader1} gives
\begin{align}
\ddot{\hat{\eta}}=&S_d\left(\hat{\omega}\right)\dot{\hat{\eta}}+S_d(\dot{\hat{\omega}})\hat{\eta}-\mu_1\left(DH\otimes I_m\right)\big(\dot{\hat{\eta}}-\dot{\bar{v}}\big). \label{noleaderd2}\tag{A.3}
\end{align}
Define the following Lyapunov function for \eqref{noleaderd2}
  $$X_1(\dot{\hat{\eta}})=\dot{\hat{\eta}}^T\left(B\otimes I_m\right)\dot{\hat{\eta}}.$$
The time derivative of $X_1(t)$ along the trajectory of \eqref{noleaderd2} is
\begin{align*}
\dot{X}_1
= &- \mu_1\dot{\hat{\eta}}^T\left(\bar{H}\otimes I_m\right)\dot{\hat{\eta}}+2\dot{\hat{\eta}}^T\left(B\otimes I_m\right)S_d(\dot{\hat{\omega}})\hat{\eta}\nonumber\\
&+2\mu_1\dot{\hat{\eta}}^T\left(BDH\otimes I_m\right)\dot{\bar{v}}\nonumber\\
\leq &- \mu_1\lambda_{h}\|\dot{\hat{\eta}}\|^2+2\|\dot{\hat{\eta}}\|\|B\|\|\dot{\hat{\omega}}\|\|\hat{\eta}\|\nonumber\\
 &+2\mu_1\|B\|\|DH\|\|\dot{\bar{v}}\|\|\dot{\hat{\eta}}\|.
\end{align*}
Using equations \eqref{noleaderd3} and \eqref{eqev} gives
 \begin{align*}
\dot{X}_1\leq &- \mu_1\lambda_{h}\|\dot{\hat{\eta}}\|^2+2\mu_2\|\dot{\hat{\eta}}\|\|B\|\|\phi_d(\hat{\eta})\left(DH\otimes I_m\right)\tilde{\eta}\|\|\hat{\eta}\|\nonumber\\
 &+2\mu_1\|B\|\|DH\|\|\dot{\bar{v}}\|\|\dot{\hat{\eta}}\|\nonumber
\nonumber\\
\leq &- \mu_1\lambda_{h}\|\dot{\hat{\eta}}\|^2+2\mu_2\|\dot{\hat{\eta}}\|\|B\|\|DH\|\|\tilde{\eta}\|\|\hat{\eta}\|^2\nonumber\\
&+2N\mu_1\|B\|\|DH\|\|\dot{v}\|\|\dot{\hat{\eta}}\| .
 \end{align*}
Let $\gamma=\frac{3\lambda_{h}}{4b_M}$, $q^{*}=N\|B\|\|DH\|$, and
\begin{align*}
  \rho(t)=&N\|B\|\|DH\|\left[\mu_2(\|v\|\|\hat{\eta}\|^2+\|\hat{\eta}\|^3)+\|\dot{v}\|\right]\nonumber\\ \leq& \rho_M\triangleq q_{\dot{\eta}}\|{v}(0)\|,\;\forall t\geq t_0,
\end{align*}
where $q_{\dot{\eta}}=N\|B\|\|DH\|\left[\mu_2 q_{{\eta}}^2\|v(0)\|(1+\|v(0)\|q_{{\eta}})+\|\omega\|\right]$. 
Then, for $\mu_1>1$, by using $X_1(\dot{\hat{\eta}})\leq  b_M \|\dot{\hat{\eta}}\|^2$, we have
 \begin{align*}
\dot{X}_1 
\leq &-\mu_1\lambda_{h}\|\dot{\hat{\eta}}\|^2+2\mu_1\|\dot{\hat{\eta}}\|\rho(t),\nonumber\\
\leq&-\mu_1\gamma X_1 + \frac{4\mu_1}{\lambda_{h}}\rho^2(t).
 \end{align*}
By the Comparison Lemma \cite[Lemma 3.4]{khalil2002nonlinear},  $X_1(t)$ satisfies the inequality
\begin{align*}
X_1(t)\leq &e^{-\mu_1\gamma (t-t_0)}X_1(t_0)+\frac{4\mu_1}{\lambda_{h}}\int_{t_0}^{t}\rho^2(\tau)e^{-(t-\tau)\mu_1\gamma}d\tau.
\end{align*}
Since,  under Assumptions \ref{Ass-S} and \ref{Ass-graph}, both $\hat{\eta}$ and $v$ are uniformly bounded independent of $\mu_2$ and $t$,
$\rho^2(t)\leq \rho_M^2$ over $t\geq t_0$.
Thus,
\begin{align}
X_1(t)
\leq &e^{-\mu_1\gamma (t-t_0)}X_1(t_0)+\frac{4\rho_M^2}{\gamma\lambda_{h}},\nonumber
\end{align}
 together with $ b_m \|\dot{\hat{\eta}}\|^2\leq X_1(\dot{\hat{\eta}})$, which further implies that
$$\|\dot{\hat{\eta}}(t)\|^2\leq  e^{-\mu_1\gamma (t-t_0)}\frac{b_M}{b_m}\|\dot{\hat{\eta}}(t_0)\|^2+\frac{4\rho_M^2}{\gamma\lambda_{h}b_m},$$
i.e., $\dot{\hat{\eta}}(t)$ is uniformly bounded for $t\geq 0$. Clearly, for any $\mu_1>1$, there exist a $t_1$, such that $$\|\dot{\hat{\eta}}(t)\|^2\leq \frac{8\rho_M^2}{\gamma\lambda_{h}b_m},$$ which is regardless of $\mu_1$ over $t\geq t_1$.

Since $v(t)$, $\dot{\hat{\eta}}(t)$ and $\hat{\eta}(t)$ are all uniformly bounded for $t\geq 0$, which further implies $S_{d}(\hat{\omega}(t))\hat{\eta}(t)$ is also uniformly bounded for $t \geq 0$ from equation \eqref{noleader1}.
By Lemma 3.3 in \cite{WangHuang2019IJC-coop}, from equation \eqref{creeq2abbbb}, we have
\begin{align}\label{SWETa}
S_{d}&(\hat{\omega}(t))\hat{\eta}(t)\nonumber\\
=&-\phi^T_{d}(\hat{\eta}(t))\hat{\omega}(t)\nonumber\\
=&-\mu_2\phi^T_{d}(\hat{\eta}(t))\int_{t_0}^{t}\phi_d\left(\hat{\eta}(\tau)\right)\left(DH\otimes I_{m}\right)\big[\hat{\eta}(\tau)-\bar{v}(\tau)\big]d\tau\nonumber\\
&-\phi^T_{d}(\hat{\eta}(t))\hat{\omega}(t_0)\tag{A.4}
\end{align}
where $\phi_d\left(\hat{\eta}\right)=\textnormal{block diag} \left(\phi(\hat{\eta}_{1}),\cdots,\phi(\hat{\eta}_{N})\right)$.
Equation \eqref{SWETa} is independent of $\mu_1$. Besides, $\hat{\eta}$ and $\bar{v}$ are uniformly bounded, regardless of $\mu_1$ and $t$. Hence, it can be concluded from \eqref{SWETa} that $S_d(\hat{\omega}(t))\hat{\eta}(t)$ is also uniformly bounded, regardless of $\mu_1$ and $t$.

\section{Proof of Lemma \ref{etape}} \label{Appendix-B}

Define the following Lyapunov function for \eqref{leadererror1},
$$X_{2}(\tilde{\eta})=\tilde{\eta}^T\left(B\otimes I_m\right)\tilde{\eta}, $$
where $B=\textnormal{block diag}(b_1,\dots,b_N)$ is a positive diagonal matrix defined in Lemma \ref{Lm-BDH}.
The time derivative of $X_2(t)$ along the trajectory of \eqref{leadererror1} is
\begin{align}\label{lyanierro}
\dot{X}_2
=&2\tilde{\eta}^T\left(B\otimes S(\omega)- \mu_1\left(BDH\otimes I_m\right)\right)\tilde{\eta}\nonumber\\
&+2\tilde{\eta}^T\left(B\otimes I_m\right)S_d\left(\tilde{\omega}\right)\hat{\eta}\nonumber\\
=&-\mu_1\tilde{\eta}^T \left(\bar{H}\otimes I_m\right)\tilde{\eta}+2\hat{\eta}^T\left(B\otimes I_m\right)S_d\left(\tilde{\omega}\right)\hat{\eta}\nonumber\\
&-2\bar{v}^T\left(B\otimes I_m\right)S_d\left(\tilde{\omega}\right)\hat{\eta}\nonumber\\
=&-\mu_1\tilde{\eta}^T \left(\bar{H}\otimes I_m\right)\tilde{\eta}-2\bar{v}^T\left(B\otimes I_m\right)S_d\left(\tilde{\omega}\right)\hat{\eta}\nonumber\\
\leq&-\mu_1\lambda_{h}\|\tilde{\eta}\|^2+2b_{M}N\|v(0)\|\|S_d\left(\tilde{\omega}\right)\hat{\eta}\|\nonumber\\
\leq&-\mu_1\delta X_2+\rho_1(t),\tag{B.1}
\end{align}
where $\delta=\frac{\lambda_{h}}{b_{M}}$ and $\rho_1(t)=2b_{M}N\|v(0)\|\|S_d\left(\tilde{\omega}\right)\hat{\eta}\|$ with $b_m$ and $b_M$ are the smallest and biggest eigenvalues of $B$, respectively.
Since
$$S_d\left(\tilde{\omega}\right)\hat{\eta}= S_d\left(\hat{\omega}\right)\hat{\eta}-\left(I_N\otimes S\left(\omega\right) \right)\hat{\eta},$$
by Lemma \ref{Lm-unif-bund-eta}, under Assumptions \ref{Ass-S} and \ref{Ass-graph},
$S_d\left(\tilde{\omega}\right)\hat{\eta}$ is  uniformly bounded, regardless of $\mu_1$ and $t$. Thus, for any initial condition $v(0)$, $\hat{\eta}_i(0)$ and $\hat{\omega}_i(0)$, we have
$ \rho_1(t) \leq\rho_1^{*}$, where $$ \rho_1^{*}=\max\nolimits_{\hat{\eta}_i(0)\in \mathds{V}_0,~\hat{\omega}_i(0)\in \mathds{W}}\left\{\rho_1(t)\right\}$$ is regardless of $\mu_1$ and $t$. From equation \eqref{lyanierro}, we have
\begin{align}
\dot{X}_2\leq-\mu_1\delta X_2+ \rho_1^{*}.\nonumber
\end{align}
By  the Comparison Lemma \cite[Lemma 3.4]{khalil2002nonlinear},  $X_2(t)$ satisfies the inequality
\begin{align}
X_2(t)&\leq e^{-\mu_1\delta t}X_2(0)+ \rho_1^{*}\int_{0}^{t}e^{-(t-\tau)\mu_1\delta}d\tau.\nonumber
\end{align}
Since $b_m \|\tilde{\eta}\|^2\leq X_2(\tilde{\eta})$, we have
\begin{align}
 b_m\|\tilde{\eta}(t)\|^2&\leq X_2(t)\leq e^{-\mu_1\delta t}X_2(0)+\frac{ \rho_1^{*}}{\mu_1\delta}.\nonumber
\end{align}
Thus, $\lim\limits _{t\rightarrow \infty}\|\tilde{\eta}(t)\|^2\leq\lim\limits _{t\rightarrow \infty}\frac{X_2(t)}{b_m}\leq\frac{ \rho_1^{*}}{\mu_1\delta b_m}$.
Hence, we have
$$\lim\limits _{t \rightarrow \infty}\|\tilde{\eta}_i(t)\|\leq \alpha(\mu_1), \quad i=1,\dots,N,$$
where $ \alpha(\mu_1)= \frac{ \sqrt{\rho_1^{*}}}{\sqrt{\mu_1\delta b_m}}$. Since  $\lim\limits_{\mu_1\rightarrow\infty}\alpha(\mu_1)=0$,  there exist $T_0> 0$ such that, for $t\geq T_0$,
$$\|\tilde{\eta}_i\|\leq 2\alpha(\mu_1),~~~~i =1\cdots,N.$$
Then, for $i=1,\cdots,N$, we have, $\forall t\geq T_0$,
\begin{align}\left(\phi(\tilde{\eta}_i)\phi^T(\tilde{\eta}_i)\right)^{\frac{1}{2}}
=\left[
                                                                 \begin{matrix}
                                                                  A_1(\tilde{\eta}_{i}) &  \cdots &   0\\
                                                                     \vdots& \ddots &  \vdots \\
                                                                    0 &  \cdots &  A_l(\tilde{\eta}_{i})\\
                                                                 \end{matrix}\right] \leq 2\alpha(\mu_1)I_{l},\nonumber\end{align}
where $A_k^2(\tilde{\eta}_{i})=\tilde{\eta}_{i,2k-1}^2+\tilde{\eta}_{i,2k}^2$,  $k=1,\cdots,l$. Since $\forall x,y\in \mathds{R}^{m}$, $\|x-y\|\geq\|x\|-\|y\|$,  for $i=1,\cdots,N$, and $k = 1, \cdots, l$, we have
\begin{align}
A_k(\tilde{\eta}_{i})&=A_k(\hat{\eta}_{i}-v)\geq A_k(v)-A_k(\hat{\eta}_{i}).\nonumber
\end{align}
 Thus, $\forall t\geq T_0$, we have
\begin{align}2\alpha(\mu_1)I_{l}\geq &\left(\phi(\tilde{\eta}_i)\phi^T(\tilde{\eta}_i)\right)^{\frac{1}{2}}\nonumber\\
\geq &\left(\phi(v)\phi^T(v)\right)^{\frac{1}{2}}-\left(\phi(\hat{\eta}_i)\phi^T(\hat{\eta}_i)\right)^{\frac{1}{2}}.\nonumber\end{align}
 Simple calculation shows that,
for $k = 1, \cdots, l$,
$$\left[
    \begin{array}{c}
       v_{2k-1}(t) \\
       v_{2k}(t) \\
    \end{array}
  \right]
=\left[
    \begin{array}{c}
      C_k\sin (\omega_{0k} t + \upsilon_k)\\
       C_k\cos (\omega_{0k} t + \upsilon_k)\\
    \end{array}
  \right],
$$
where $C_k = \sqrt{v^2_{2k-1} (0) + v^2_{2k} (0)}$ and $ \tan \upsilon_k =  \frac{v_{2k-1} (0)}{v_{2k} (0)}$. Let $C_{\min}=\min\left\{C_1,\cdots,C_l\right\}$. Since,  for $i = 1, \cdots, l$,   $\col (v_{2i-1} (0), v_{2i} (0)) \neq 0$, $C_{\min}>0$.
 In fact,
\begin{align}\label{betav0}\left(\phi(v)\phi^T(v)\right)^{\frac{1}{2}}
&=\left[
      \begin{matrix}
        C_1 &  \cdots &   0\\
         \vdots & \ddots &  \vdots \\
        0  & \cdots  & C_l \\
      \end{matrix}
    \right]\geq C_{\min} I_{l}.\tag{B.2}
\end{align}
Choose  $\mu_1>\mu_1^{*}$ with
\begin{align} \label{mu0-star}
\mu_1^{*}
&\geq \frac{ 16 \rho_1^{*}}{C_{\min}^2 \delta b_{m}},\tag{B.3}
\end{align}
such that, for
$t\geq T_0$, and $i=1,\cdots,N$,
\begin{align}\frac{C_{\min}}{2}I_l&\geq\left(\phi(\tilde{\eta}_i)\phi^T(\tilde{\eta}_i)\right)^{\frac{1}{2}}\nonumber\\
&\geq \left(\phi(v)\phi^T(v)\right)^{\frac{1}{2}}-\left(\phi(\hat{\eta}_i)\phi^T(\hat{\eta}_i)\right)^{\frac{1}{2}}.\nonumber
\end{align}
Then, for $t\geq T_0$ and $\mu_1>\mu_1^{*}$, we have
\begin{align}
\left(\phi(\hat{\eta}_i)\phi^T(\hat{\eta}_i)\right)^{\frac{1}{2}}\geq & \left(\phi(v)\phi^T(v)\right)^{\frac{1}{2}}-\frac{C_{\min}}{2}I_l\nonumber\\
\geq  &\frac{C_{\min}}{2}I_l,\nonumber
\end{align}
and thus
\begin{align} \label{eqphieta}
\phi(\hat{\eta}_i)\phi^T(\hat{\eta}_i)\geq& \frac{C_{\min}^2}{4}I_l.\tag{B.4}
\end{align}
By integrating both sides of equation \eqref{eqphieta}, we have
\begin{align}\frac{1}{T}\int^{t+T}_{t}\phi(\hat{\eta}_i(\tau))\phi^T(\hat{\eta}_i(\tau))d\tau \geq \frac{C_{\min}^2}{4}I_l.\nonumber
\end{align}
In other words, $\phi(\hat{\eta}_i(\tau))$ is PE.
According to \eqref{eqphieta}, we have that
\begin{align}\|\phi^T (\hat{\eta}_i)\hat{\omega}_i\|^2= \hat{\omega}_i^T\phi (\hat{\eta}_i)\phi^T (\hat{\eta}_i)\hat{\omega}_i
\geq  \frac{C_{\min}^2}{4} \hat{\omega}_i^T\hat{\omega}_i\nonumber
\end{align}
By Lemma 3.3 in \cite{WangHuang2019IJC-coop}, we have $S(\hat{\omega}_i)\hat{\eta}_i = -\phi^T (\hat{\eta}_i)\hat{\omega}_i$, and thus
\begin{align} \label{20210201-eq1}
\hat{\omega}_i^T\hat{\omega}_i \leq \frac{4}{C_{\min}^2} \|\phi^T (\hat{\eta}_i)\hat{\omega}_i\|^2 =\frac{4}{C_{\min}^2}  \| S(\hat{\omega}_i)\hat{\eta}_i \|^2.\tag{B.5}
\end{align}
Lemma \ref{Lm-unif-bund-eta} shows that $S(\hat{\omega}_i)\hat{\eta}_i$ is  uniformly bounded and independent of $\mu_1$, so is $\hat{\omega}_i$, according to \eqref{20210201-eq1}. The boundedness also holds for $\tilde{\omega}_i$ as $\omega$ is a constant vector and belongs to a compact set $\mathds{W}$.
Therefore, from \eqref{mu0-star} and the definition of $\rho_1^{*}$, we can choose $\mu_1^{*}$ as
\begin{align}\label{mu1-star-2} \mu_1^{*}= &\max\limits_{\hat{\eta}_i(0)\in \mathds{V}_0,~\hat{\omega}_i(0)\in \mathds{W}}\left\{\frac{32b_{M}^2N\|v(0)\|\|\tilde{\omega}\|\|\hat{\eta}\|}{C_{\min}^2b_{m}\lambda_{h}}\right\}.\tag{B.6}
\end{align}

\section{Proof of Lemma \ref{Lm-exp}} \label{Appendix-C}

Since $Y$ is diagonalizable and has real and positive eigenvalues, by Lemma \ref{matWC}, there exists a symmetric positive definite matrix $U$ such that $UY$ is symmetric positive definite.
The following development is motivated by Lemma B.2.3 in \cite{marino1996nonlinear}.
Let us first show that, for any initial condition, $\lim\limits_{t\rightarrow\infty}x(t)=0$.

Consider the radially unbounded function
\begin{align}\label{omeganormv0}
V_0 =x^T\left(UY\otimes P_a\right)x+\kappa^{-1}z^T\left(U\otimes I_s\right)z. \tag{C.1}
\end{align}
Its time derivative along system \eqref{creeq3exabbbb} is
\begin{align}
\dot{V}_0
=&x^T\left(UY\otimes (P_aA_a+A^T_aP_a)-2Y^TUY\otimes P_a\right)x \nonumber\\
&-2x^T\left(UY\otimes (P_a\psi^T(t))\right)z+2z^T\left(UY\otimes(\psi(t)P_a)\right)x. \nonumber
\end{align}
As $P_aA_a+A^T_aP_a\leq 0$ and $UY$ is positive definite, we have
\begin{align}
\label{omeganormv1}
\dot{V}_0 \leq -2x^T\left(Y^TUY\otimes P_a\right)x \leq 0.\tag{C.2}
\end{align}
Thus $V_0$ is uniformly bounded and non-increasing, and hence both $x$ and $z$ are uniformly bounded, which further implies the boundedness of $\dot{x}$.

Let $\lambda_{\min}$ be the smallest eigenvalue of $Y^T U Y\otimes P_a$. Then according to \eqref{omeganormv1}, we have
  \begin{align}\label{Chapt7attract}
 2 \lambda_{\min}\lim_{t\rightarrow \infty}\int_{t_0}^{t} \| x(\tau)\|^2 d\tau  & \leq - \lim_{t\rightarrow \infty}\int_{t_0}^{t} \dot{V}_0(\tau)d\tau \nonumber\\
&=V_0(t_0)-V_0(\infty)\leq\infty. \tag{C.3}
\end{align}

By Barbalat's lemma, we can conclude that
  \begin{align}\label{x0}\lim\limits_{t\rightarrow\infty}x(t)=0.\tag{C.4}\end{align}
Now, we are left to show that for any initial condition
 \begin{align}\label{z0}\lim\limits_{t\rightarrow\infty}z(t)=0.\tag{C.5}\end{align}
Before doing that, we first prove the following claim.  For convenience, let $\kappa=1$.
\begin{claim} \label{Claim-z}
Consider system \eqref{creeq3exabbbb}. Given any $\epsilon>0$,  and any initial condition $x(t_0)$ and $z(t_0)$, there exist $T>0$ such that $\| z(t)\| <\epsilon$, $\forall t>T$.
\end{claim}

\begin{proof}
Suppose for any $\epsilon>0$, there does not exist $t_1>0$ such that
\begin{align}\label{omegabe}
\| z(t)\|\geq \epsilon,~~~~\forall t\geq t_1.\tag{C.6}
\end{align}
Consider
\begin{align}
\varphi(t)=\frac{1}{2}\Big[z^T(t+T)\left(U\otimes I_{s}\right)z(t+T)-z^T(t)\left(U\otimes I_{s}\right)z(t)\Big]\nonumber
\end{align}
which is uniformly bounded since $\| z(t)\|$ is uniformly bounded $\forall t\geq t_0$. Then, the time derivative of $\varphi(t)$ along \eqref{creeq3exabbbb} is
\begin{align}\label{equniform}
\dot{\varphi}(t)=&z^T(t+T)\left(U\otimes I_{s}\right)\dot{z}(t+T)-z^T(t)\left(U\otimes I_{s}\right)\dot{z}(t)\nonumber\\
=&\int_{t}^{t+T}\frac{d}{d\tau}\left(z^T(\tau)\left(U\otimes I_{s}\right)\dot{z}(\tau)\right)d\tau\nonumber\\
 =&\int_{t}^{t+T}\frac{d}{d\tau}\left[z^T(\tau) \left((UY)\otimes (\psi(\tau)P_a)\right)x(\tau)\right]d\tau\nonumber\\
=&\int_{t}^{t+T}\Big[x^T(\tau)\left((Y^TUY)\otimes (P_a\psi^T(\tau)\psi(\tau)P_a)\right)\nonumber\\
&+z^T(\tau)\big((UY)\otimes (\dot{\psi}(\tau)P_a)-(Y^TUY)\otimes (\psi(t)P_a)\big)\nonumber\\
&+z^T(\tau)\left((UY)\otimes (\psi(t)P_a)\right)\left(I_n\otimes  A_a\right)\Big]x(\tau)d\tau\nonumber\\
&-\int_{t}^{t+T}z^T(\tau)\left[(UY)\otimes \left(\psi(\tau)P_a\psi^T(\tau)\right)\right]z(\tau)d\tau\nonumber\\
=& \varphi_1(t)- \varphi_2(t),\tag{C.7}
\end{align}
where $\varphi_1(t)$ and $\varphi_2(t)$ are the first and second term of equation \eqref{equniform}, respectively.
Since $ \psi(t)$ and $ \dot{\psi}(t)$ are uniformly bounded, there exists $M_\psi>0$ such that $\| \psi(t)\|\leq M_\psi$ and $\|\dot{\psi}(t)\|\leq M_\psi$, $\forall t\geq t_0$. By  \eqref{omeganormv0} and \eqref{omeganormv1}, it is straightforward to show
\begin{align}\label{boundxz}
\|x(t)\|^2&\leq{\frac{\|x_0\|^2\lambda_{UY}\lambda_{P_a}+\|z_0\|^2\lambda_{U}}{\lambda_{uy}\lambda_{p_a}}}\triangleq M_x^2,\nonumber\\
\|z(t)\|^2&\leq {\frac{\|x_0\|^2\lambda_{UY}\lambda_{P_a}+\|z_0\|^2\lambda_{U}}{\lambda_{u}}}\triangleq M_z^2,\tag{C.8}
\end{align}
 where $(x_0,z_0)$ is the initial condition, $\lambda_{u}$ and $\lambda_{U}$ are the smallest and largest eigenvalues of $U$, $\lambda_{p_a}$ and $\lambda_{P_a}$ are the smallest and largest eigenvalues of $P_a$, and $\lambda_{uy}$ and $\lambda_{UY}$ are the smallest and largest eigenvalues of $UY$.
Then,
\begin{align}
\varphi_1(t)\leq& \lambda_{UY}\lambda_{P_a}M_\psi\Big[M_\psi M_x\lambda_{P_a}\| Y\|+M_z+M_z  \| Y\|\nonumber\\
&+M_z \|A_a\|\Big]\int_{t}^{t+T}\| x(\tau)\| d\tau.\nonumber
\end{align}
On the other hand, from \eqref{x0}, there exists a time instant $t_2$ such that
\begin{align}\label{Chapt7equniform2}
\varphi_1(t)\leq \frac{\kappa}{2}\lambda_{uy}\lambda_{p_a}\epsilon^2,~~~~\forall t\geq t_2.\tag{C.9}
\end{align}
Suppose there exists a time instant $t_1$ such that \eqref{omegabe} holds. By assumption, $\psi(t)$ is PE such that
 $$\int_{t}^{t+T}\psi(\tau)P_a\psi^T(\tau)d\tau\geq\kappa \lambda_{p_a}I_{s}, $$ for some constant $\kappa>0$.
Thus we have,
$$\int_{t}^{t+T}c^T\left[(UY)\otimes \left(\psi(\tau)P_a\psi^T(\tau)\right)\right]cd\tau\geq \kappa\lambda_{p_a}\lambda_{uy},$$
$\forall t \geq t_0$, and $\forall c:\| c\|=1$, which in turn, along with  \eqref{omegabe}, implies
\begin{align}\frac{1}{\epsilon^2}&\int_{t}^{t+T}z^T(\tau)\left((UY)\otimes \left(\psi(\tau)P_a\psi^T(\tau)\right)\right)z(\tau)d\tau\nonumber\\
&\geq \int_{t}^{t+T}\frac{z^T(\tau)}{\| z(\tau)\|}\left[(UY)\otimes \left(\psi(\tau)P_a\psi^T(\tau)\right)\right]\frac{z(\tau)}{\| z(\tau)\|}d\tau\nonumber\\
&\geq \kappa\lambda_{p_a}\lambda_{uy} ,~~~~\forall t\geq t_1.\nonumber
\end{align}
Thus, we have, $\forall t\geq t_1$,
\begin{align}\varphi_2(t)=&\int_{t}^{t+T}z^T(\tau)\left[(UY)\otimes \left(\psi(\tau)P_a\psi^T(\tau)\right)\right]z(\tau)d\tau\nonumber\\
\geq & \kappa\lambda_{p_a}\lambda_{uy}\epsilon^2.\nonumber
\end{align}
From  \eqref{equniform} and \eqref{Chapt7equniform2}, we obtain
\begin{align}
\dot{\varphi}(t)\leq\varphi_1(t)- \varphi_2(t)\leq -\frac{\kappa}{2}\lambda_{p_a}\lambda_{uy}\epsilon^2,~~\forall t\geq t_3,\nonumber
\end{align}
with $t_3=\max\{t_1,t_2\}$, which contradicts the boundedness of $\varphi(t)$ for any $t\geq t_0$.
\end{proof}

Since $x(t) \to 0$, $\forall \epsilon>0$, there exists $t_{\epsilon}>0$ such that
\begin{align}\label{Chapt7equniform5}
\| x(t)\|^2\leq \frac{{\lambda_{u}}\epsilon}{{2\lambda_{UY}\lambda_{P_a}}}, ~~\forall t\geq t_{\epsilon}.\tag{C.10}
\end{align}
By Claim \ref{Claim-z}, there exists $T_{\epsilon}\geq t_{\epsilon}$ such that
\begin{align}\label{Chapt7equniform6}
\| z(T_{\epsilon})\|^2\leq \frac{{\lambda_{u}}\epsilon}{{2\lambda_{U}}}.\tag{C.11}
\end{align}
For initial conditions $\tilde{\eta}(T_{\epsilon})$ and $\tilde{\omega}(T_{\epsilon})$,
according to \eqref{boundxz}, \eqref{Chapt7equniform5} and \eqref{Chapt7equniform6},
we have $\forall t\geq T_{\epsilon}$,
 $$\| z(t)\|^2\leq {\frac{\| x(T_{\epsilon})\|^2\lambda_{UY}\lambda_{P_a}+\| z(T_{\epsilon})\|^2\lambda_{U}}{\lambda_{u}}}\leq\epsilon, $$
i.e., $z(t) \to 0$ as $t \in \infty$.
Since \eqref{Chapt7attract} holds uniformly with respect to $t_0$, so $\lim\limits_{t\rightarrow\infty}x(t)=0$ and $\lim\limits_{t\rightarrow\infty}z(t)=0$ uniformly.
  It follows that the equilibrium point at the origin is globally uniformly asymptotically stable. Since system \eqref{creeq3exabbbb} is a linear time varying system, by the Theorem 4.11 in \cite{khalil2002nonlinear}, the equilibrium is also exponentially stable.

\end{appendices}

\bibliographystyle{ieeetr}
\bibliography{TAC}

\begin{IEEEbiography}[{\includegraphics[width=1in,height=1.25in,clip,keepaspectratio]{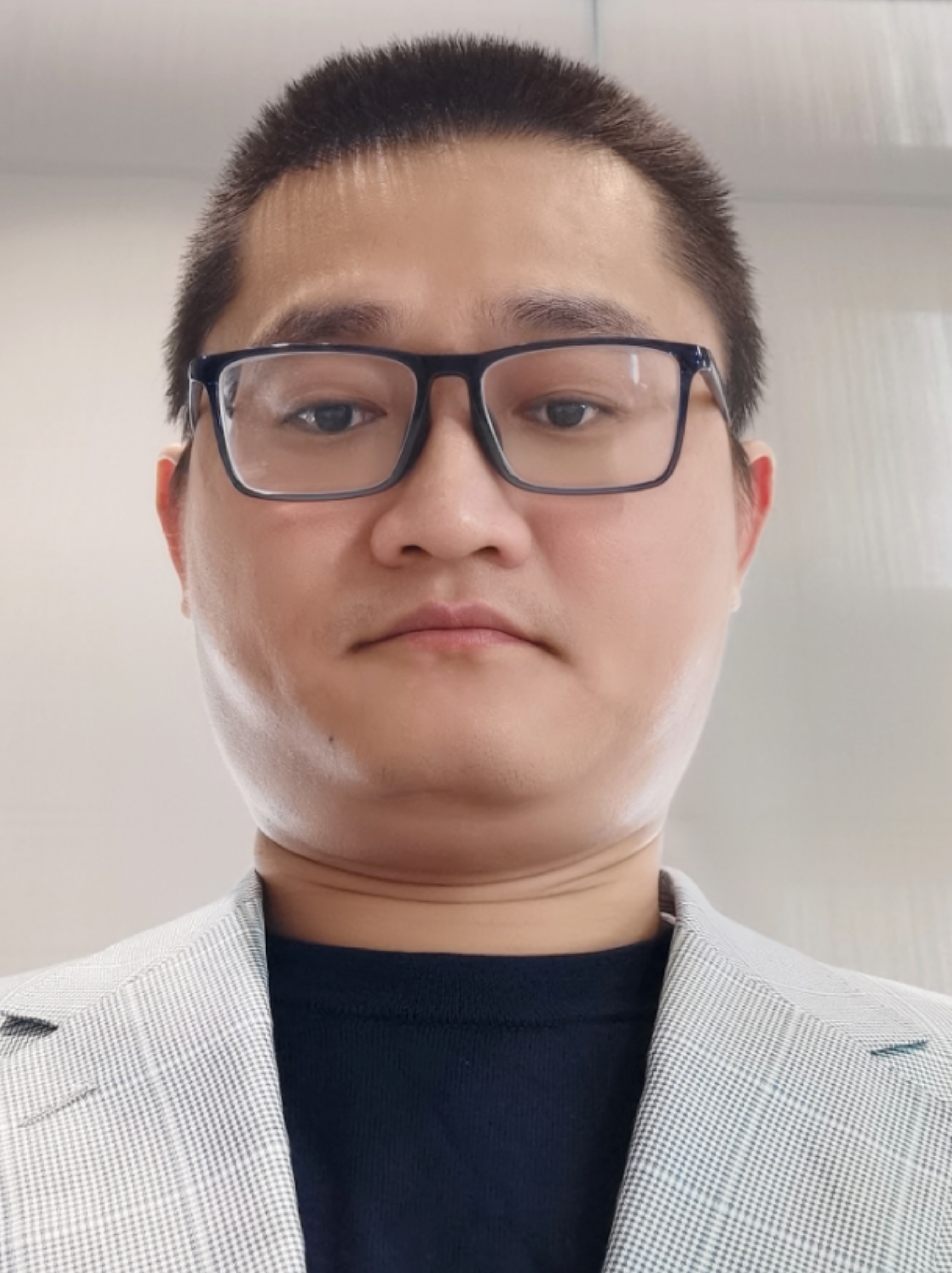}}]{Shimin Wang}
received the B.Sci. degree and M.Eng. degree from Harbin
Engineering University, Harbin, China, in 2014. He then received Ph.D. degree from The Chinese University of Hong Kong, Hong Kong, China, in 2019. From 2014 to 2015, he was an engineer in the Jiangsu Automation Research Institute.  He is now working as a postdoctoral fellow with the Department of Electrical and Computer Engineering, University of Alberta, Edmonton, AB, Canada.
\end{IEEEbiography}

\begin{IEEEbiography}[{\includegraphics[width=1in,height=1.25in,clip,keepaspectratio]{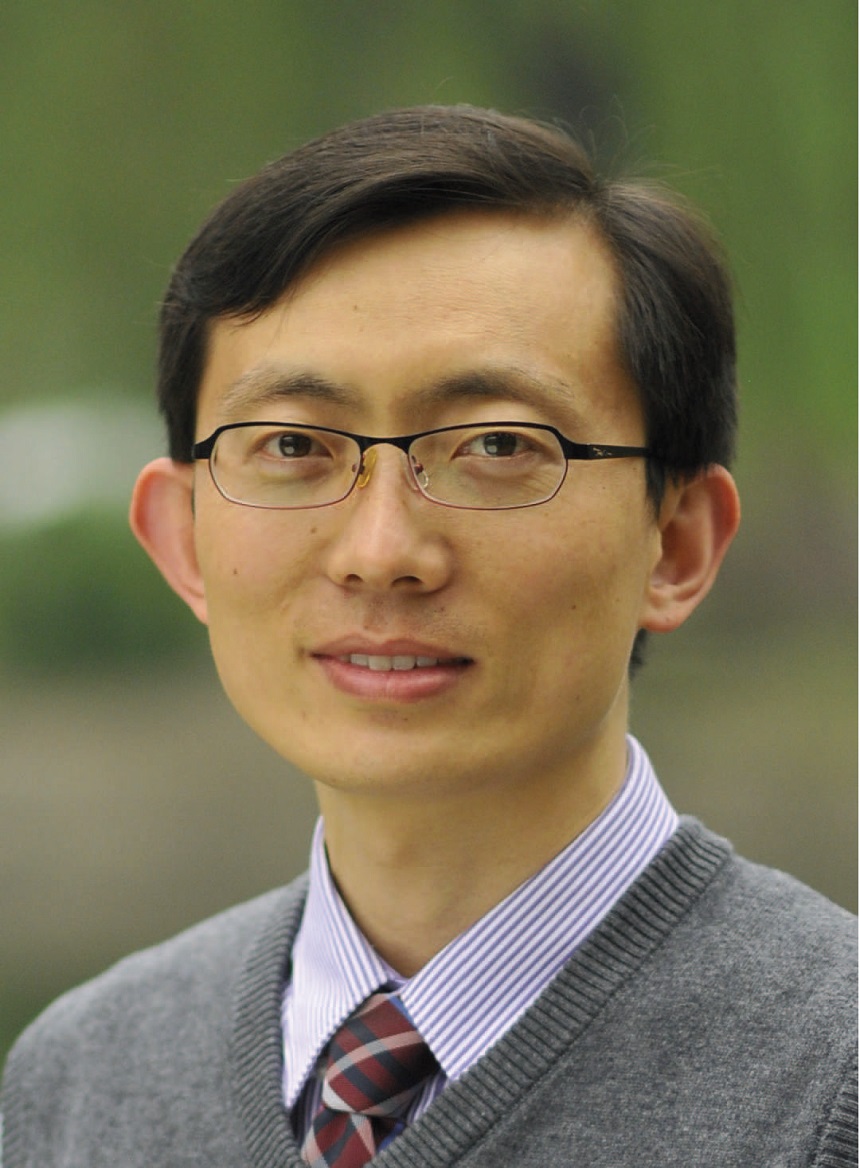}}]{Hongwei Zhang}
 received the B.S. and M.S. degrees from Tianjin University in 2003
and 2006, respectively; and the Ph.D. degree from the Chinese
 University of Hong Kong in 2010. Subsequently, he held postdoctoral positions at the University of
Texas at  Arlington and the City University of Hong Kong. He had been working at the Southwest Jiaotong University from 2012 to 2020. He joined the Harbin Institute of Technology, Shenzhen, in Nov. 2020, where he is now a
Professor.  His research
interests include cooperative control of multi-agent systems, neural adaptive control, nonlinear control, and active
noise control. He is an Associate Editor of Neurocomputing and Transactions of the Institute of Measurement and Control.
\end{IEEEbiography}

\begin{IEEEbiography}[{\includegraphics[width=1in,height=1.25in,clip,keepaspectratio]{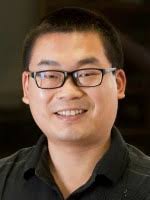}}]{Zhiyong Chen} received the B.E. degree from the University of Science and Technology of China, and the M.Phil. and Ph.D.
degrees from the Chinese University of Hong Kong, in 2000, 2002 and 2005, respectively. He worked as a Research
Associate at the University of Virginia during 2005--2006. He joined the University of Newcastle, Australia, in 2006,
where he is currently a Professor. He was also a Changjiang Chair Professor with Central South University, Changsha,
China. His research interests include non-linear systems and control, biological systems, and multi-agent systems. He
is/was an associate editor of Automatica, IEEE Transactions on Automatic Control and IEEE Transactions on Cybernetics.\end{IEEEbiography}

\end{document}